\documentclass[lettersize,journal]{IEEEtran}
\usepackage{cite}
\usepackage[utf8]{inputenc} % allow utf-8 input
\usepackage[T1]{fontenc}    % use 8-bit T1 fonts
\usepackage{hyperref}       % hyperlinks
\usepackage{url}            % simple URL typesetting
\usepackage{booktabs}       % professional-quality tables
\usepackage{amsfonts}       % blackboard math symbols
\usepackage{nicefrac}       % compact symbols for 1/2, etc.
\usepackage{microtype}      % microtypography
\usepackage{doi}
\usepackage{amsmath,amssymb,amsfonts}
\usepackage{bbding}
\usepackage{algorithmic}
\usepackage{graphicx}
\usepackage{textcomp}
\usepackage{xcolor}
\usepackage{bbm}
\usepackage{amsthm}
\usepackage{setspace}
\usepackage{physics}
\usepackage{algorithm}
\usepackage{mathtools}
\usepackage{diagbox}
\usepackage{caption}
\usepackage{subcaption}
\usepackage{physics}
\usepackage{multirow}
\usepackage{epstopdf}
\usepackage{tabularx}
\makeatletter
\newcommand{\distas}[1]{\mathbin{\overset{#1}{\kern\z@\sim}}}%
\newsavebox{\mybox}\newsavebox{\mysim}
\newcommand{\distras}[1]{%
  \savebox{\mybox}{\hbox{\kern3pt$\scriptstyle#1$\kern3pt}}%
  \savebox{\mysim}{\hbox{$\sim$}}%
  \mathbin{\overset{#1}{\kern\z@\resizebox{\wd\mybox}{\ht\mysim}{$\sim$}}}%
}
\makeatother
\makeatletter
\newcommand*\rel@kern[1]{\kern#1\dimexpr\macc@kerna}
\newcommand*\widebar[1]{%
  \begingroup
  \def\mathaccent##1##2{%
    \rel@kern{0.8}%
    \overline{\rel@kern{-0.8}\macc@nucleus\rel@kern{0.2}}%
    \rel@kern{-0.2}%
  }%
  \macc@depth\@ne
  \let\math@bgroup\@empty \let\math@egroup\macc@set@skewchar
  \mathsurround\z@ \frozen@everymath{\mathgroup\macc@group\relax}%
  \macc@set@skewchar\relax
  \let\mathaccentV\macc@nested@a
  \macc@nested@a\relax111{#1}%
  \endgroup
}
\makeatother
\newtheorem{definition}{Definition}
\newtheorem{assumption}{Assumption}
\newtheorem{theorem}{Theorem}
\newtheorem{lemma}{Lemma}
\newtheorem{remark}{Remark}

\newtheorem{corollary}{Corollary}

\newcommand{\etal}{\textit{et al. }}
\begin{document}
\title{Anomaly Detection in Networked Bandits}
\author{Xiaotong Cheng and Setareh Maghsudi\thanks{
X. Cheng and S. Maghsudi are with the Department of Electrical Engineering and Information Technology, Ruhr-University Bochum, 44801 Bochum, Germany (email:xiaotong.cheng@ruhr-uni-bochum.de, setareh.maghsudi@rub.de).

A part of this paper is accepted for presentation at the 35th IEEE International Workshop on Machine Learning for Signal Processing (IEEE MLSP 2025).}\\
}

\markboth{Journal of \LaTeX\ Class Files,~Vol.~14, No.~8, August~2021}%
{}

\maketitle
%--------------------------------------------------------->Abstract
\begin{abstract}
The nodes' interconnections on a social network often reflect their dependencies and information-sharing behaviors. Nevertheless, abnormal nodes, which significantly deviate from most of the network concerning patterns or behaviors, can lead to grave consequences. Therefore, it is imperative to design efficient online learning algorithms that robustly learn users' preferences while simultaneously detecting anomalies. 

We introduce a novel bandit algorithm to address this problem. Through network knowledge, the method characterizes the users' preferences and residuals of feature information. By learning and analyzing these preferences and residuals, it develops a personalized recommendation strategy for each user and simultaneously detects anomalies. We rigorously prove an upper bound on the regret of the proposed algorithm and experimentally compare it with several state-of-the-art collaborative contextual bandit algorithms on both synthetic and real-world datasets.
\end{abstract}
%-------------------------------------------------------------

%---------------------------------------->Keywords
%------------------------------------------------------------>Introduction
\section{Introduction}
\IEEEPARstart{M}{ulti}-armed bandit (MAB) is a prototypical solution for many applications, such as online advertising \cite{li2010exploitation}, personalized recommendation \cite{mahadik2020fast}, etc. In the classical MAB setting, a decision-maker selects one action at each round and receives an instantaneous reward, aiming to find policies that maximize its cumulative rewards over all rounds. 

We focus on a powerful MAB setting known as contextual bandits, where side information about users or items is available. In particular, each arm (action) has a feature vector known as ``context'', and the expected payoff depends on the unknown bandit parameters and the given context. Contextual bandits offer significant advantages in large-scale applications \cite{wu2016contextual} and have become a reference model for adaptive recommender systems \cite{li2010exploitation,bouneffouf2012contextual}. Recently, several researchers have studied contextual bandits in collaborative settings \cite{cesa2013gang,wu2016contextual,yang2020laplacian}, where a network models the information-sharing and dependency among users. The potential affinities between pairs of users are beneficial to improving the action selection policy.  

The state-of-the-art research assumes that the network structure accurately reflects users' similarities or relationships. That assumption ignores the impact of anomalies/outliers in the network. Anomalies in a network are nodes that do not conform to the patterns of the majority \cite{li2017radar}. Integrating the network structure information into the learning algorithm on one hand can effectively accelerate the learning process at each node. On the other hand, the existence of anomalies, if not handled properly, can easily mislead the collaborative system and significantly corrupt the algorithm's performance. Thus, it is imperative to incorporate anomaly detection into online learning algorithms to mitigate the negative impact caused by these anomalies. 

To better explain the practical relevance of anomaly detection, we offer concrete examples across various real-world scenarios. In social networks, anomalies often correspond to fake or bot users who exhibit interaction patterns that diverge significantly from typical user behavior, e.g., forming random or overly dense connections, or spreading misinformation \cite{10909361}. In recommendation systems, anomalous behavior may stem from malicious users who inject biased profiles to manipulate item rankings, commonly known as shilling attacks \cite{zhou2014detection}. Anomaly detection also plays a critical role in other areas, including scientific discovery~\cite{grun2015single,chaudhary2015folding}, fraud detection in financial systems~\cite{bolton2001unsupervised,kshetri2010economics}, and medical screening or diagnosis~\cite{schiff2017screening}, where identifying rare or abnormal patterns can be crucial for decision-making. 

In summary, anomaly detection aims to identify noteworthy items, events, or entities that do not conform to the expected patterns of the majority in a dataset \cite{chandola2009anomaly,ding2019interactive}. Various algorithms are proposed for anomaly detection; Nonetheless, they typically suffer from some limitations: Some do not fit sequential analysis or online learning \cite{laxhammar2013online}; some rely on the assumption that properties of anomalies are known in advance, which might not be true \cite{li2017radar}; and many involve tuning several hyperparameters such as anomaly thresholds, which might lead to overfitting or poorly-calibrated alarm rates  \cite{laxhammar2013online}. 

We address anomaly detection in networked contextual bandits, where anomalies refer to users whose behaviors or underlying feature vectors significantly deviate from those expected under the mutual influence structure of the network. Motivated by \cite{li2017radar}, we focus on residual analysis, where anomalies manifest as outliers in the deviation between observed and estimated collaborative user behavior. By characterizing the residuals of users' features and leveraging the coherence between users' features and network information, the anomalies can be detected. We further extend this idea to an online setting within the networked contextual bandit framework, enabling real-time anomaly detection and decision-making. Our main contributions are as follows.
%--------------------------------------->Contributions
\begin{itemize}
    \item We are the first to formally formulate the networked bandits with anomalies problem, which extends traditional models by explicitly incorporating the possibility of anomalies within the networked structure. Our formulation removes the common assumption that the network structure perfectly mirrors user interactions, making it more applicable to real-world applications where anomalies exist and may interfere with collaborative learning.
    \item To address the formulated problem, we propose an algorithm called NEtworked LinUCB with Anomaly detection (NELA), which integrates two key functions: Learning users' preferences and detecting anomalies in a networked system. This dual-function approach allows NELA to robustly estimate the bandit parameters for all users, utilizing network information to efficiently reduce the sample complexity for preference learning and minimize regret. Based on the estimated global bandit parameters, NELA accurately detects anomalies. To the best of our knowledge, NELA is the first framework to address these two problems concurrently. 
    \item We rigorously prove the regret of NELA is upper bounded by $O(dn\log(1+\frac{\sum_{t=1}^{T}\sum_{j=1}^n \boldsymbol{W}(u_t,j)^2}{dn \lambda_1})\sqrt{T}+ s_As_1\tau + \log s_0\sqrt{s_0T})$ and also provide an $\Omega(\xi d\sqrt{nT} + \sqrt{s_0T})$ regret lower bound with $\xi \in [1,\sqrt{n}]$. In addition, we prove the correctness of NELA in anomaly detection theoretically.
    \item Extensive experiments on both synthetic and real-world datasets demonstrate the effectiveness and superior performance of our proposed algorithm compared to the state-of-the-art methods.
\end{itemize}

The remainder of the paper is organized as follows. Section~\ref{sec:related-work} reviews the related work. In Section~\ref{sec:problem}, we formulate the anomaly detection problem in networked bandits. Section~\ref{sec:alg} presents the details of the proposed algorithm, Networked LinUCB with Anomaly Detection (NELA). Theoretical analysis is discussed in Section~\ref{sec:analysis}. To evaluate the effectiveness of the proposed algorithm, we conduct experiments on both synthetic and real-world datasets, as described in Section~\ref{sec:experiment}. Finally, we conclude our work and highlight the directions of future work in Section~\ref{sec:conclusion}.

%--------------------------------------------------------->Selected Work
\section{Related Work}\label{sec:related-work}
Previous research actively explores exploiting users' social relationships to accelerate bandit learning. To adaptively learn and utilize unknown user relations,   \cite{gentile2014online} and \cite{li2016collaborative} initially propose clustering of bandits problems, which, however, suffer from the irreversible splitting of clusters. Reference \cite{li2019improved} introduces an improved algorithm that uses sets to represent clusters, allowing split and merge actions at each time step. In addition, authors in \cite{cheng2023parallel} solve the clustering of bandits problem with a game-theoretic approach. 

Some researchers use social network topologies to enhance algorithm performance. Reference \cite{cesa2013gang} combines contextual bandits with the social graph information, proposing the GOB algorithm that allows signal-sharing among users by exploiting the network structure. Similarly, Wu \etal \cite{wu2016contextual} investigate scenarios when neighboring users’ judgments influence each other through shared information. Reference \cite{yang2020laplacian} improves GOB by using the random-walk graph Laplacian to model the signal smoothness among the feature vectors of users and addresses the high scaling complexity problem of GOB. 

Another line of related work is anomaly/outlier detection. Several papers study outlier detection in the MAB framework; Nonetheless, they focus on either identifying arms whose expected rewards deviate significantly from the majority \cite{zhuang2017identifying,ban2020generic,zhu2020robust} or selecting a single anomalous node or link for a human expert to evaluate \cite{ding2019interactive,meng2021interactive,zhang2023multi}. These approaches are typically limited to addressing a single task at a time. In contrast, the challenge of identifying outliers in multi-user systems during interaction through users' action selection remains less explored. Besides, some anomaly detection research focuses on attributed networks \cite{perozzi2014focused,liu2017accelerated,li2017radar,peng2018anomalous}. In \cite{perozzi2014focused}, the authors infer user preferences. They then simultaneously form clusters and find outliers according to the inferred user preferences. Reference \cite{li2017radar} proposes a learning framework to model the residuals of attribute information and its coherence with network information for anomaly detection.
%------------------------------------------------------->Problem Formulation
\section{Problem Formulation}\label{sec:problem}
\textbf{Notations:} We denote vectors by bold lowercase letters, matrices by bold uppercase letters, and sets or events by calligraphic fonts. The notation $[\cdot]$ represents a set of integers within a specific range (e.g., $[d]$ represents $\{1, \ldots, d\}$). We represent the $i$-th row of matrix $\boldsymbol{A}$ as $\boldsymbol{A}(i,:)$, the $j$-th column as $\boldsymbol{A}(:,j)$ and the $(i,j)$-th entry as $\boldsymbol{A}(i,j)$. $\norm{\boldsymbol{x}}_p$ denotes the $p-$norm of a vector $\boldsymbol{x}$ and the weighted $2-$norm of vector $\boldsymbol{x}$ is defined by $\norm{\boldsymbol{x}}_{\boldsymbol{A}} = \sqrt{\boldsymbol{x}^{\top}\boldsymbol{A}\boldsymbol{x}}$. The inner product is denoted by $\langle \cdot, \cdot \rangle$. $vec(\cdot)$ denotes the vectorize operation. \textbf{Table~\ref{tab:nota}} summarizes important definitions.

%----------------Table
\begin{table}[!ht]
\begin{center}
\centering
\captionsetup{justification=centering}
\caption{Notation}
\label{tab:nota}
 \begin{tabular}{c p{7cm}}
 \hline
 \multicolumn{2}{l}{\textbf{Problem-specific notations}} \\
 \hline
 $n$ & Number of users \\
 $\mathcal{G}$ & The user connected graph \\
 $\mathcal{U}$ & The set of users  \\
 $\tilde{\mathcal{U}}$ & The set of anomalies \\
 $\mathcal{E}$ & The set of edges \\
 $\boldsymbol{W}$ & Adjacency matrix of user connected graph \\
 $d$ & Dimension of feature vectors \\
 $T$ & Total number of rounds \\
 $\boldsymbol{\theta}_i$ & Feature vector of user $i$ \\
 $\boldsymbol{v}_i$ & Residual feature vector of user $i$ \\
 $\mathcal{A}_{i,t}$ & Set of context vectors of user $i$ at time $t$ \\
 %$M$ & Number of context vectors in set $\mathcal{A}_{i,t}$ \\
 $\boldsymbol{x}_{a_t,u_t}$ & The action of user $u_t$ at time $t$ \\
 $r_{u_t,a_t}$ & Payoff of user $u_t$ at time $t$ \\
 $\boldsymbol{\Theta}$ & Global feature parameter matrix \\
 $\boldsymbol{V}$ & Global residual feature parameter matrix \\
 $\mathring{\boldsymbol{X}}_{a_t,u_t}$ & Only the column corresponding to user $u_t$ is set to $\boldsymbol{x}_{a_t,u_t}^T$, the other columns are set to zero \\
 $\mathcal{J}$ & Support of $vec(\boldsymbol{V})$ (see Remark~\ref{rem:s}) \\
 $s_0$ & Sparsity index of $vec(\boldsymbol{V})$ (see Remark~\ref{rem:s}) \\
 $\boldsymbol{\Sigma}$ & Theoretical expected $d \times d$ covariance matrix \\
 $\widehat{\Sigma}_{i,t}$ & Empirical covariance matrix for user $i$ \\
 $\mathcal{T}_{i,t}$ & Time steps that user $i$ is selected until $t$ \\
 $\overline{\boldsymbol{\Sigma}}$ & Multi-user theoretical expected covariance matrix \\
 $\widehat{\overline{\boldsymbol{\Sigma}}}$ & Multi-user empirical covariance matrix \\
 $\lambda_t$ & Regularizer for Lasso at round $t$ \\
 $\lambda_0$ & Coefficient of the regularizer for Lasso \\
 $\hat{\mathcal{J}}_0^t$, $\hat{\mathcal{J}}_1^t$ & Estimate of the support after the first and second thresholding, respectively \\
 $\widehat{\boldsymbol{V}}_t$ & Estimated vector of $\boldsymbol{V}$ \\
 \hline
 \end{tabular}
 \end{center}
\end{table}
%----------------

We consider a known network graph, $\mathcal{G} = \{\mathcal{U}, \mathcal{E}\}$, where $\mathcal{U} = \{1,2,\ldots, n\}$ denotes the set of $n$ users and $\mathcal{E}$ represents the set of edges encoding affinities between users. The set of abnormal users is denoted by $\tilde{\mathcal{U}} \subseteq \mathcal{U}$. We represent the network structure using a weighted matrix $\boldsymbol{W} \in \mathbb{R}^{n \times n}$. Similar to the settings in \cite{wu2016contextual}, each element $\boldsymbol{W}(i,j)$ of $\boldsymbol{W}$ is non-negative and proportional to the influence that user $i$ has on user $j$ in determining the payoffs. $\boldsymbol{W}$ is column-wise normalized such that $\sum_{i=1}^n\boldsymbol{W}(i,j) = 1$ for $j \in \mathcal{V}$. We assume $\boldsymbol{W}$ is time-invariant and known to the learner beforehand. 

Each user $i \in \mathcal{U}$ has an unknown feature vector $\boldsymbol{\theta}_i \in \mathbb{R}^d$. We collect these vectors in the global feature parameter matrix $\boldsymbol{\Theta} = (\boldsymbol{\theta}_1, \ldots, \boldsymbol{\theta}_n) \in \mathbb{R}^{d\times n}$. At each time step $t = 1,2, \ldots, T$, a user $u_t \in \mathcal{U}$ appears uniformly at random. Given a set $\mathcal{A}_{t} \subset \mathbb{R}^d$ consisting of $M$ arms, the learner selects an arm $\boldsymbol{x}_{t}$ to recommend to $u_t$, and obtains a payoff $r_t$. If $\mathcal{F}_t$ is the $\sigma$-algebra generated by random variables $(\mathcal{A}_1,\boldsymbol{x}_1,r_1,\ldots,\mathcal{A}_{t-1},\boldsymbol{x}_{t-1},r_{t-1},\mathcal{A}_t)$, $\boldsymbol{x}_t$ is $\mathcal{F}_t$-measurable. 

Previous work considers only the mutual influence among users \cite{wu2016contextual,wang2017factorization}, resulting in a payoff $r_t = \boldsymbol{x}_t^\top\boldsymbol{\Theta}\boldsymbol{W}(:,u_t) + \eta_t$, where $\eta_{t} \sim \mathcal{N}(0,\sigma^2)$ is a zero mean sub-Gaussian random variable with variance $\sigma^2$ given $\mathcal{F}_t$. However, such assumptions may not hold in the presence of anomalies, since anomalous behaviors do not conform to the patterns of the majority. To measure users' feature vectors' deviation from network-based collaborative model $\boldsymbol{\Theta}\boldsymbol{W}$, we associate each user $i \in \mathcal{U}$ with an unknown residual feature vector $\boldsymbol{v}_i \in \mathbb{R}^d$. The global residual parameter matrix yields $\boldsymbol{V} = (\boldsymbol{v}_1, \ldots, \boldsymbol{v}_n) \in \mathbb{R}^{d\times n}$. We define the bandit parameter matrix as $\overline{\boldsymbol{\Theta}} = \boldsymbol{\Theta}\boldsymbol{W} + \boldsymbol{V}$, which determines the reward. This additive-deviation formulation models anomalies as unusual individual behaviors that deviate from the group-level influence pattern, following the approach in \cite{zhao2016online,li2017radar,peng2018anomalous} and we extend it to the online setting. Under this assumption, the payoff for user $u_t$'s jointly determined by (i) mutual influences from other users through $\boldsymbol{\Theta}\boldsymbol{W}$ and and (ii) the user’s personal residual feature vector $\boldsymbol{V}$
\begin{align}
r_{a_t,u_t} &= \boldsymbol{x}_{a_t,u_t}^{\top}\overline{\boldsymbol{\Theta}}(:,u_t) + \eta_{t}  \notag \\   
&= vec(\mathring{\boldsymbol{X}}_{a_t,u_t})^\top vec(\boldsymbol{\Theta}\boldsymbol{W}+\boldsymbol{V}) + \eta_{t}. \label{eq:payoff}
\end{align}
In \eqref{eq:payoff}, $\mathring{\boldsymbol{X}}_{a_t,u_t} = (\boldsymbol{0}, \ldots, \boldsymbol{x}_{a_t,u_t}, \ldots, \boldsymbol{0}) \in \mathbb{R}^{d \times n}$ denotes the matrix, where the column corresponding to user $u_t$ is $\boldsymbol{x}_{a_t,u_t}^{\top}$, and the rest is zero.

The learner intends to learn a selection strategy that minimizes the cumulative regret with respect to an optimal strategy. Formally, the cumulative regret yields
\begin{align}
    R(T) = \mathbb{E}\Big[\sum_{t=1}^T R_t \Big] = \mathbb{E}\Big[ \sum_{t=1}^{T} (r_{a_t^*,u_t} - r_{a_t,u_t}) \Big],
\end{align}
where $u_t$ is the user served at $t$ and $r_{a_t^*,u_t}$ is the corresponding optimal payoff with respect to user $u_t$. The user $u_t$ is sampled uniformly at random from the entire user set. To support the primary goal of regret minimization, accurate anomaly detection becomes a necessary secondary objective. Specifically, at each round $t$, the learner will also select a set of users $\tilde{\mathcal{U}}_t$ as detected anomalies, and the objective is to make $\tilde{\mathcal{U}}_t$ as close to the ground-truth
set of anomalies $\tilde{\mathcal{U}}$ as possible. Formally, the secondary objective is two-fold: i) guarantee that all true anomalies are included, i.e., $\tilde{\mathcal{U}} \subseteq \tilde{\mathcal{U}}_t$; ii) minimize the number of normal users mistakenly identified as anomalies, i.e., bound $|\tilde{\mathcal{U}}_t \setminus \tilde{\mathcal{U}}|$.

In the following, we state some assumptions on feature and residual vectors. 

%------------------------------->Assumption
\begin{assumption}[Anomaly Residual Feature Vector]\label{assump:arfv}
The deviation of node $i$'s behavior from the expected behavior modeled by the network-based global model $\boldsymbol{\Theta}\boldsymbol{W}$ is represented by $\boldsymbol{v}_i$. In the absence of anomalies, users are expected to follow the mutual influence structure encoded by the graph, and hence their residual vectors vanish, i.e. $\norm{\boldsymbol{v}_i} = 0$. In contrast, a large residual norm indicates that the user's actions deviate significantly from what the network would predict, indicating anomalous behavior. Therefore, we assume a user to be anomalous if $\norm{\boldsymbol{v}_i} > \gamma > 0$. This formulation is inspired by \cite{li2017radar}, where anomalies are modeled as violations of low-rank structure in an offline matrix setting. Our model explicitly considers the residual deviation of feature vectors, making it compatible with online contextual bandit settings. 
\end{assumption}
%------------------------------------
%------------------------------->Remark
\begin{remark}[Jointly Sparsity]
\label{rem:s}
Since anomalies occur less frequently than normal users, the global residual feature vector $vec(\boldsymbol{V}) \in \mathbb{R}^{nd}$ is a sparse vector. Define $\mathcal{J}(\boldsymbol{v}_i) = \{j \in [d], \boldsymbol{v}_i(j) \neq 0 \}$ as the sparsity pattern with respect to single-user residual vector $\boldsymbol{v}_i$, representing the subset of $|\mathcal{J}(\boldsymbol{v}_i)|$ non-zero components of vector $\boldsymbol{v}_i \in \mathbb{R}^d$. There is a set $\mathcal{J} \subset [nd]$ such that $\mathcal{J}(\boldsymbol{v}_i) \subset \mathcal{J}, \quad  \forall i \in \mathcal{U}$.
\end{remark}
%------------------------------
%-------------------------------
\begin{assumption}\label{assump:bound}
Let $s_0 = |\mathcal{J}|$ denote the joint sparsity of the global residual vector $vec(\boldsymbol{V})$, where $s_0 \ll nd$. Under the sparsity assumption, the minimum nonzero entry satisfies $vec(\boldsymbol{V})_{\min} \geq \gamma/\sqrt{d}$. We further assume that the residual vector is bounded in both $\ell_1$ and $\ell_2$ norms: $\norm{vec(\boldsymbol{V})}_1 \leq s_1$ and $\norm{vec(\boldsymbol{V})} \leq s_v$. In addition, we assume the context vectors are bounded in $\ell_2$ and $\ell_{\infty}$ norms: for all $t$ and all $\boldsymbol{x} \in \mathcal{A}_{u_t,t}$, $\norm{\boldsymbol{x}}_2 \leq s_x$, and $\norm{\boldsymbol{x}}_{\infty} \leq s_A$. Finally, we assume the global parameter vector is also bounded: $\norm{vec(\boldsymbol{\Theta})} \leq s_{\theta}$.
\end{assumption}
%--------------------

%----------------------------------->Algorithm
\section{Algorithm}\label{sec:alg}
In this section, we introduce the NELA algorithm, which jointly learns users' feature vectors while simultaneously estimating their residuals to detect anomalies. Algorithm~\ref{alg:cola} presents the pseudocode. 

Accurate estimation of the global feature vector is critical for minimizing regret. Moreover, since anomalies are encoded in the residual component $\boldsymbol{V}$, accurately estimating $\boldsymbol{V}$ serves as a sufficient condition for effective anomaly detection. NELA algorithm maintains a profile $(\boldsymbol{A}_t,\boldsymbol{b}_t)$ for the global feature vector $\boldsymbol{\Theta}$ estimation and tracks the estimation of $\boldsymbol{V}$ for anomaly detections. With the expected payoff in \eqref{eq:payoff}, it is natural to deploy regularized regression to estimate the global feature matrix $\boldsymbol{\Theta}$ and residual matrix $\boldsymbol{V}$ at round $t$ as follows,
\begin{align}
    \min_{\boldsymbol{\Theta},\boldsymbol{V}} \frac{1}{2}\sum_{s=1}^t & (vec(\mathring{\boldsymbol{X}}_{a_s,u_s})^{\top} vec(\boldsymbol{\Theta}\boldsymbol{W}+\boldsymbol{V}) -  r_{a_s,u_s})^2  \notag \\
    & + \frac{\lambda_1}{2}\norm{vec(\boldsymbol{\Theta})}^2_2 + \frac{\lambda_2}{2}\norm{vec(\boldsymbol{V})}^2_2, \label{eq:ridge-reg}
\end{align}
where $\lambda_1$ and $\lambda_2$ is parameters of $l_2$ regularization in ridge regression. The closed-form solutions of \eqref{eq:ridge-reg} are easy to obtain via the alternating least square method with ridge regression. However, this can result in a term $O(\sqrt{ndT})$ in regret that scales polynomially with $vec(\boldsymbol{V})$'s dimension $nd$. In modern systems with thousands of users, this significantly increases the regret and the number of interactions required before anomalies are detected. To address the full dimensionality-dependence term in regret, we further reformulate the objective function \eqref{eq:ridge-reg} considering the sparse property of $vec(\boldsymbol{V})$ as follows
\begin{align}
    \min_{\boldsymbol{\Theta},\boldsymbol{V}} \frac{1}{2}\sum_{s=1}^t & (vec(\mathring{\boldsymbol{X}}_{a_s,u_s})^{\top} vec(\boldsymbol{\Theta}\boldsymbol{W}+\boldsymbol{V}) -  r_{a_s,u_s})^2  \notag \\
    & + \frac{\lambda_1}{2}\norm{vec(\boldsymbol{\Theta})}^2_2 + \frac{\lambda_s}{2}\norm{vec(\boldsymbol{V})}_1. \label{eq:ridge-lasso-reg}
\end{align}
Therefore, by carefully selecting parameter $\lambda_s$, $s \in [t]$ and applying lasso estimator, the regret term $O(\sqrt{ndT})$ can be reduced to $\mathcal{O}(\log (s_0)\sqrt{s_0T})$ with $s_0 \ll nd$. We solve the optimization problem \eqref{eq:ridge-lasso-reg} using a widely employed alternating method: fixing one parameter and solving for the other iteratively with closed-form solutions \cite{wang2017factorization,liu2018transferable}. The closed-form solutions are given by  $vec(\widehat{\boldsymbol{\Theta}}_t) = \boldsymbol{A}_t^{-1}\boldsymbol{b}_t$ and $vec(\widehat{\boldsymbol{V}}_t) = \arg \min_{vec(\boldsymbol{V})}\norm{\mathcal{R} - \mathcal{X}_Jvec(\boldsymbol{V})}_2^2$, where the detailed computation of $(\boldsymbol{A}_t,\boldsymbol{b}_t,\mathcal{R},\mathcal{X}_J)$ can be found in Algorithm~\ref{alg:cola} (line 8). 

At round $t$, a user $u_t$ appears to be severed and the arm set $\mathcal{A}_{u_t,t}$ arrives for the learner to choose from. The learner chooses the arm for user $u_t$ according to the following strategy
\begin{align}
    a_t = \arg \max_{a \in \mathcal{A}_{u_t,t}} & \Big (vec(\mathring{\boldsymbol{X}}_{a,u_{t}})^{\top} vec(\widehat{\boldsymbol{\Theta}}_t\boldsymbol{W} + \widehat{\boldsymbol{V}}_t) \notag \\
    &\quad  + \alpha_t \norm{vec(\mathring{\boldsymbol{X}}_{a,u_{t}}\boldsymbol{W}^{\top})}_{\boldsymbol{A}_{t-1}^{-1}} \Big ), \label{eq:action}
\end{align}
where $\alpha_t$ is a parameter that characterizes the upper bound for the estimation of $vec(\widehat{\boldsymbol{\Theta}}_t)$. See Lemma~\ref{lem:ucb} for more details. Specifically, the solution for $\widehat{\boldsymbol{\Theta}}_t$ is obtained using a ridge-regression approach with a fixed $\widehat{\boldsymbol{V}}_t$ and $\widehat{\boldsymbol{V}}_t$ is determined using a thresholded bandit algorithm \cite{ariu2022thresholded} with a fixed $\widehat{\boldsymbol{\Theta}}_t$.

To reduce false alarms in anomaly detection, NELA incorporates a two-step thresholding mechanism. After obtaining an initial estimate $\widehat{\boldsymbol{V}}_0$ of the residual matrix $\boldsymbol{V}$, each component is compared against a predefined threshold to produce a preliminary estimate of the support of $\boldsymbol{V}$. A second thresholding step with a larger threshold is then applied to further eliminate potential false positives by filtering out remaining normal users. This two-stage process enhances both the accuracy of anomaly detection and the estimation quality of the residual vector. 

In fact, anomaly detection and residual estimation in NELA are not isolated tasks, but mutually reinforcing processes. Detecting anomalies corresponds to identifying the support set of the residual vector $\boldsymbol{V}$. A more accurate estimate of this support set, denoted by $\hat{\mathcal{J}}_1^t$, leads to a sparser and more targeted estimation of the residual vector $vec(\widehat{\boldsymbol{V}}_t)$. Conversely, a more precise estimate of $\boldsymbol{V}$ provides a clearer signal for support detection, enhancing the reliability of anomaly identification in subsequent rounds. NELA explicitly makes use of this synergy, iteratively refining both the residual estimation and the detected anomaly set $\tilde{\mathcal{U}}_t$, resulting in low false alarm rates and improved learning performance.

%------------------------------>Algorithm
\begin{algorithm}[!ht]
\caption{Networked LinUCB with Anomaly detection (NELA)} 
\label{alg:cola}
\begin{algorithmic}[1]
\STATE \textbf{Input}: Graph $\mathcal{G} = \{\mathcal{U}, \mathcal{E}\}$, parameters $\lambda_0$, $\lambda_1$ and $\alpha_t$. 
\STATE \textbf{Initialization} 
\begin{itemize}
    \item $\boldsymbol{A}_0 = \lambda_1 \boldsymbol{I} \in \mathbb{R}^{nd\times nd}$, $\boldsymbol{b}_0 = \boldsymbol{0} \in \mathbb{R}^{nd}$. 
    \item $vec(\widehat{\boldsymbol{\Theta}}_0) = \boldsymbol{A}_0^{-1}\boldsymbol{b}_0$ and $vec(\widehat{\boldsymbol{V}}_0) = \boldsymbol{0} \in \mathbb{R}^{nd}$
    \item $\tilde{\mathcal{U}}_t = \Phi$, $\forall t$
\end{itemize}
\FOR{$t = 1,2, \cdots, T$}
\STATE Receive user $u_t$
\STATE Observe context vectors $\boldsymbol{x}_a$ for each item $a \in \mathcal{A}_{u_t,t}$
\STATE Select item according to \eqref{eq:action}
\STATE Observe the payoff $r_{a_t,u_t}$
\STATE Update
\begingroup
\small
\begin{align*}
    &\boldsymbol{A}_t \leftarrow \boldsymbol{A}_{t-1} + vec(\mathring{\boldsymbol{X}}_{a_t,u_t}\boldsymbol{W}^{\top})vec(\mathring{\boldsymbol{X}}_{a_t,u_t}\boldsymbol{W}^{\top})^{\top} \\
    &\boldsymbol{b}_t \leftarrow \boldsymbol{b}_{t-1} + vec(\mathring{\boldsymbol{X}}_{a_t,u_t}\boldsymbol{W}^{\top}) \times \\
    &\quad \quad \quad \quad (r_{a_t,u_t} - vec(\mathring{\boldsymbol{X}}_{a_t,u_t})^{\top}vec(\widehat{\boldsymbol{V}}_t)) \\
    &vec(\widehat{\boldsymbol{\Theta}}_t) = \boldsymbol{A}_t^{-1}\boldsymbol{b}_t \\
    &\mathcal{X} \leftarrow (vec(\mathring{\boldsymbol{X}}_{a_1,u_1}), vec(\mathring{\boldsymbol{X}}_{a_2,u_2}),\ldots, vec(\mathring{\boldsymbol{X}}_{a_t,u_t}))^{\top},\\
    &\mathcal{R} \leftarrow (r_{a_1,u_1}- vec(\mathring{\boldsymbol{X}}_{a_1,u_1})^{\top}vec(\widehat{\boldsymbol{\Theta}}_1\boldsymbol{W}),\ldots, \\
    & \quad \quad  \quad r_{a_t,u_t}- vec(\mathring{\boldsymbol{X}}_{a_t,u_t})^{\top}vec(\widehat{\boldsymbol{\Theta}}_t\boldsymbol{W})) \\
    &\lambda_t \leftarrow \lambda_0 \sqrt{\frac{2\log t \log(nd)}{t}} \\
    &\widehat{\boldsymbol{V}}_0^t \leftarrow \arg \min_{\boldsymbol{V}}\frac{1}{t}\norm{\mathcal{R}- \mathcal{X}\boldsymbol{V}}_2^2 + \lambda_t\norm{\boldsymbol{V}}_1 \\
    &\hat{\mathcal{J}}_0^t \leftarrow \{j \in [nd]: |\widehat{\boldsymbol{V}}_0^t(j)| > 4\lambda_t\} \\
    &\hat{\mathcal{J}}_1^t \leftarrow \{j \in \hat{\mathcal{J}}_0^t: |\widehat{\boldsymbol{V}}_0^t(j) |> 4\lambda_t\sqrt{|\hat{\mathcal{J}}_0^t|}\} \\
    &\mathcal{X}_J \leftarrow (vec(\mathring{\boldsymbol{X}}_{a_1,u_1})(\hat{\mathcal{J}}_1^t),\ldots, vec(\mathring{\boldsymbol{X}}_{a_t,u_t})(\hat{J}_1^t))^{\top} \\
    &vec(\widehat{\boldsymbol{V}}_t) \leftarrow \arg \min_{vec(\boldsymbol{V})}\norm{\mathcal{R} - \mathcal{X}_Jvec(\boldsymbol{V})}_2^2
\end{align*}
\endgroup
\STATE Identify user $i$ as anomalous and add it to $\tilde{\mathcal{U}}_t$, if $j \in \hat{J}_1^t$ and $id \leq j \leq (i+1)d$. 
\ENDFOR
\end{algorithmic}
\end{algorithm}
%-----------------------------

%---------------------->Theoretical Analysis
\section{Theoretical Analysis}\label{sec:analysis}
In this section, we present a theoretical analysis of the proposed algorithm, NELA. We begin by outlining the key assumptions required for our analysis. Based on these assumptions, we then establish performance guarantees for NELA in terms of regret bounds and anomaly detection.

%------------------>Assumptions
\subsection{Assumptions}
We use the following assumptions throughout the analysis. Specially, Assumptions~\ref{assump:context-d}, \ref{assump:armsd}, and \ref{assump:spd} are standard assumptions for sparse linear bandits \cite{oh2021sparsity,ariu2022thresholded}, with adaptation to a multi-user setting. Definition~\ref{def:cc} guarantees the compatibility conditions in a multi-user setting, similar to the RSC condition in \cite{cella2023multi}. 

%---------------------------------->Assumption
\begin{assumption}[Context Distribution]
\label{assump:context-d}
At round $t$, the arm set $\mathcal{A}_{u_t,t} \subset \mathbb{R}^d$ consists of $M$ $d$-dimensional vectors $\boldsymbol{x}_m$. The tuples $\mathcal{A}_{u_1,1}, \ldots, \mathcal{A}_{u_t,t}$ are i.i.d. and follow a fixed unknown zero mean sub-Gaussian joint distribution $p$ on $\mathbb{R}^{Md}$ \cite{oh2021sparsity,ariu2022thresholded}. 
\end{assumption}
%------------------------------
%--------------------------------->Assumption
\begin{assumption}[Arm Distribution]
\label{assump:armsd}
There exists a constant $\nu < \infty$, such that $p(-\bar{\boldsymbol{x}})/p(\bar{\boldsymbol{x}}) \leq \nu$, $\forall \bar{\boldsymbol{x}} \in \mathbb{R}^{Md}$. Moreover, there exists a constant $\omega_{\mathcal{X}} < \infty$ such that, for any permutation $(a_1,a_2,\ldots,a_M)$ of $[M]$, any integer $m \in \{2, \ldots, M-1\}$ and any fixed $\boldsymbol{v}$, it holds \cite{oh2021sparsity,ariu2022thresholded}
\begin{align}
    \mathbb{E}&[\boldsymbol{x}_{a_m}\boldsymbol{x}_{a_m}^{\top}\mathbb{I}\{\boldsymbol{x}_{a_1}^{\top}\boldsymbol{v} < \ldots < \boldsymbol{x}_{a_M}^{\top}\boldsymbol{v}\}] \preceq \notag \\
    & \omega_{\mathcal{X}} \mathbb{E}[(\boldsymbol{x}_{a_1}\boldsymbol{x}_{a_1}^{\top}+\boldsymbol{x}_{a_M}\boldsymbol{x}_{a_M}^{\top})\mathbb{I}\{\boldsymbol{x}_{a_1}^{\top}\boldsymbol{v} < \ldots < \boldsymbol{x}_{a_M}^{\top}\boldsymbol{v}\}] \notag 
\end{align}
\end{assumption}
%------------------------------
Parameter $\nu$ in Assumption~\ref{assump:armsd} characterizes the skewness of the arm distribution $p$; for symmetrical distributions $\nu = 1$. Assumption~\ref{assump:armsd} is satisfied for a large class of distributions, e.g., Gaussian, multi-dimensional uniform, and Bernoulli \cite{ariu2022thresholded,cella2023multi}. The value $\omega_{\mathcal{X}}$ captures dependencies between arms, the more positively correlated they are, the smaller $\omega_{\mathcal{X}}$ will be. When arms are generated i.i.d. from a multi-variate Gaussian or multi-variate uniform distribution over the sphere we have $\omega_{\mathcal{X}} = O(1)$ \cite{ariu2022thresholded,cella2023multi}. 

We further define the covariance matrices for multi-user systems. First, we define the theoretical averaged $d \times d$ covariance matrix as 
\begin{gather}
    \boldsymbol{\Sigma} = \frac{1}{M} \mathbb{E}[\sum_{m=1}^M\boldsymbol{x}_m\boldsymbol{x}_m^{\top}] = \frac{1}{M} \sum_{m=1}^M \boldsymbol{\Sigma}_m,
\end{gather}
where the expectation is over the decision set sampling distribution $p$, which is assumed to be shared between users. For every $i \in \mathcal{U}$, denote the empirical covariance matrix as
\begin{gather}
    \widehat{\boldsymbol{\Sigma}}_{i,t} = \frac{1}{|\mathcal{T}_{i,t}|}\sum_{ s \in \mathcal{T}_{i,t}}\boldsymbol{x}_{i,s}\boldsymbol{x}_{i,s}^{\top},
\end{gather}
where $\mathcal{T}_{i,t}$ is the set collects the rounds that user $i$ is served. We use the notation $\overline{\boldsymbol{\Sigma}}$ and $\widehat{\overline{\boldsymbol{\Sigma}}}_t \in \mathbb{R}^{dn \times dn}$ for the theoretical and empirical multi-user matrices respectively. They are both block diagonal and composed of the corresponding $n$ single user $d \times d$ matrices on the diagonal, that is
\begin{align}
    &\overline{\boldsymbol{\Sigma}} = diag(\boldsymbol{\Sigma}, \ldots, \boldsymbol{\Sigma}) \\
    &\widehat{\overline{\boldsymbol{\Sigma}}}_t =  diag(\widehat{\boldsymbol{\Sigma}}_{1,t}, \ldots, \widehat{\boldsymbol{\Sigma}}_{n,t}).
\end{align}
%
%------------------------->Definition
\begin{definition}[Compatibility Conditions]\label{def:cc}
Let $\overline{\boldsymbol{\Sigma}} \in \mathbb{R}^{nd \times nd}$ be the multi-user covariance matrix and $\mathcal{J} \subset [nd]$. We say that $\mathcal{J}$ satisfies the compatibility conditions with constant $\phi(\overline{\boldsymbol{\Sigma}}, \mathcal{J}) > 0$, if $\forall \boldsymbol{\Delta} \in \mathbb{R}^{nd}\backslash  \{\boldsymbol{0}\}$ such that $\sum_{j \in \mathcal{J}^c}\norm{\boldsymbol{\Delta}(j)} \leq 3\sum_{j \in \mathcal{J}}\norm{\boldsymbol{\Delta}(j)}$, where $\mathcal{J}^c$ is the complement of $\mathcal{J}$, it holds \cite{cella2023multi}
\begin{gather}
    \norm{\boldsymbol{\Delta}(\mathcal{J})}^2 \leq \frac{\norm{\boldsymbol{\Delta}}^2_{\overline{\boldsymbol{\Sigma}}}}{\phi^2(\overline{\boldsymbol{\Sigma}}, \mathcal{J})}.
\end{gather}
We assume that $\phi_0^2 > 0$ is the lower bound of $\phi^2(\overline{\boldsymbol{\Sigma}}, \mathcal{J})$. 
\end{definition}
%------------------------------------

%------------------------------>Assumption
\begin{assumption}[Sparse Positive Definiteness]\label{assump:spd}
Define for each $\mathcal{B} \subset [nd]$, $\overline{\boldsymbol{\Sigma}}_B = \frac{1}{M}\mathbb{E}[\sum_{m=1}^M\boldsymbol{X}_m(\mathcal{B})\boldsymbol{X}_m(\mathcal{B})^{\top}]$, where $\boldsymbol{X}_m$ is a $nd$-dimensional vector that satisfying $\overline{\boldsymbol{\Sigma}} = \frac{1}{M}\mathbb{E}[\sum_{m=1}^M\boldsymbol{X}_m\boldsymbol{X}_m^{\top}]$. All the expectations here are over the decision set sampling distribution $p$. $\boldsymbol{X}_m(\mathcal{B})$ is a $|\mathcal{B}|$-dimensional vector, which is extracted from the elements of $\boldsymbol{X}_m$ with indices in $\mathcal{B}$. There exists a positive constant $\beta > 0$ such that $\forall \mathcal{B} \subset [nd]$ \cite{ariu2022thresholded},
\begin{align}
    |\mathcal{B}| \leq &  s_0 + (4\nu\omega_{\mathcal{X}}\sqrt{s_0})/\phi_0^2 \quad \text{and} \quad \mathcal{J} \subset \mathcal{B} \notag \\
    & \rightarrow \min_{\boldsymbol{l} \in \mathbb{R}^{|B|}: \norm{\boldsymbol{l}}_2 = 1} \boldsymbol{l}^{\top} \overline{\boldsymbol{\Sigma}}_B \boldsymbol{l} \geq \beta,
\end{align}
where the parameters $\phi_0$, $\nu$, $\omega_{\mathcal{X}}$ are defined in Assumptions~\ref{assump:context-d},\ref{assump:armsd} and Definition~\ref{def:cc}. 
\end{assumption}
%----------------------------

Assumption~\ref{assump:spd} ensures that the context distribution is sufficiently diverse to allow reliable estimation of the residual vectors corresponding to anomalies. This assumption is specifically leveraged in analyzing the regression performance over the low-dimensional subspace defined by the active residual indices $\hat{\mathcal{J}}_1^t$. We also emphasize that Assumptions~\ref{assump:arfv}-\ref{assump:spd} are introduced solely to facilitate theoretical analysis and may not strictly hold in empirical settings.

\subsection{Performance Guarantee}
First, we prove the estimation error of $\boldsymbol{\Theta}$ is bounded with high probability in the following Lemma~\ref{lem:ucb}. 
%------------------------------
\begin{lemma}\label{lem:ucb}
With probability at least $1-\delta$, it holds
\begingroup
\small
\begin{align*}
    \norm{vec(\widehat{\boldsymbol{\Theta}}_{t}) - vec(\boldsymbol{\Theta})}_{\boldsymbol{A}_t} &\leq (\sigma + 2s_xs_v)\sqrt{\log(\frac{\det(\boldsymbol{A}_t)}{\lambda_1^{dn}\delta^2})} + \sqrt{\lambda_1}s_{\theta}.
\end{align*}
\endgroup
We use $\alpha_t$ to denote the upper bound of $\norm{vec(\widehat{\boldsymbol{\Theta}}_{t}) - vec(\boldsymbol{\Theta})}_{\boldsymbol{A}_t}$ for all $t \in [T]$. 
\end{lemma}
%---------------------------------------------
\begin{proof}
    See Appendix~\ref{app:lem-ucb}.
\end{proof}
%----------------------------------
Then, we demonstrate the performance guarantee in terms of anomaly detection in the following Lemma~\ref{lem:support-recover}. 

%-------------------------->Lemma
\begin{lemma}[Estimation of the Support \& Anomaly Detection]\label{lem:support-recover}
Let $t \geq \frac{2\log(2n^2d^2)}{C_0^2}$ such that $4\Big(\frac{4\nu\omega_{\mathcal{X}}s_0}{\phi_0^2}+ \sqrt{(1+\frac{4\nu\omega_{\mathcal{X}}}{\phi_0^2})s_0}\Big)\lambda_t \leq \gamma/\sqrt{d} \leq vec(\boldsymbol{V})_{\min}$, where $vec(\boldsymbol{V})_{\min}$ refers to the smallest non-zero entry. Under Assumptions~\ref{assump:arfv}-\ref{assump:armsd}, the detected anomaly set includes all anomalous users and the false alarm is bounded with high probability, i.e., $\mathbb{P}(\tilde{\mathcal{U}} \subset \tilde{\mathcal{U}}_t$ and $|\tilde{\mathcal{U}}_t \backslash \tilde{\mathcal{U}}| \leq \frac{4\nu\omega_{\mathcal{X}}\sqrt{s_0}}{\phi_0^2}) \geq 1 - 2\exp(-\frac{t\lambda_t^2}{32\sigma^2s_A^2}+\log nd)-\exp(-\frac{tC_0^2}{2})$. 
\end{lemma}
%------------------------
\begin{proof}
    See Appendix~\ref{app:lem-support-recover}. 
\end{proof}
%------------------------------

Based on previous Lemmas, we can derive an upper bound on the instantaneous regret. 
%--------------------->Lemma
\begin{lemma}\label{lem:ins-rgt}
Under Assumptions~\ref{assump:arfv}-\ref{assump:spd}, define event $\mathcal{D}_t = \{\mathcal{J} \subset \hat{\mathcal{J}}_1^t \quad  and \quad  
 |\hat{\mathcal{J}}_1^t \backslash \mathcal{J}| \leq \frac{4\nu\omega_{\mathcal{X}}\sqrt{s_0}}{\phi_0^2}\}$ and event $\mathcal{L}_t^{\frac{\beta}{4\nu\omega_{\mathcal{X}}}} = \Big\{ \lambda_{\min}(\widehat{\overline{\boldsymbol{\Sigma}}}_{t,\hat{\mathcal{J}}}) \geq \frac{\beta}{4\nu\omega_{\mathcal{X}}} \Big\}$, where $\widehat{\overline{\boldsymbol{\Sigma}}}_{t,\hat{J}} = \frac{1}{t}\sum_{s=1}^t vec(\mathring{\boldsymbol{X}}_{a_s,u_s})(\hat{\mathcal{J}})vec(\mathring{\boldsymbol{X}}_{a_s,u_s})(\hat{\mathcal{J}})^{\top}$ is the empirical Gram matrix on the estimated support with $\hat{\mathcal{J}}_1^t = \hat{\mathcal{J}}$ for simplicity. For any $t$, the instantaneous regret of NELA is bounded by
\begin{align}
    &\mathbb{E}[R_t] \leq \frac{4\nu\omega_{\mathcal{X}}}{\beta}2s_xs_{\theta}\sqrt{\frac{2}{t}\log(2d/\delta)}   \notag \\
    &\quad+ \frac{14h_0^2\sigma s_A\nu \omega_{\mathcal{X}}}{\beta}\sqrt{\frac{2(s_0+\frac{4\nu\omega_{\mathcal{X}}}{\phi_0^2})}{t-1}} \notag \\
    &\quad  + \alpha_t\norm{vec(\mathring{\boldsymbol{X}}_{a_t,u_t}\boldsymbol{W}^{\top})}_{\boldsymbol{A}_{t-1}^{-1}} + \alpha_t  \norm{vec(\mathring{\boldsymbol{X}}_{a_t^*,u_t}\boldsymbol{W}^{\top})}_{\boldsymbol{A}_{t-1}^{-1}} \notag \\
    &\quad  + 2s_1s_A\Big (\mathbb{P}(\mathcal{D}_t^c) + \mathbb{P}((\mathcal{L}_t^{\frac{\beta}{4\nu\omega_{\mathcal{X}}}})^c | \mathcal{D}_t) \Big). 
\end{align}
The probability of events $\mathcal{D}_t$ and $\mathcal{L}_t^{\frac{\beta}{4\nu\omega_{\mathcal{X}}}}$ are associated with Lemma~\ref{lem:support-recover} and Lemma~\ref{lem:min-eigen} and ~\ref{lem:eat} in Appendix~\ref{app:aux}.
\end{lemma}
%--------------------------------
\begin{proof}
    See Appendix~\ref{app:lem-ins-rgt}.
\end{proof}
%------------------

Define $\tau = \max \Big(\lfloor \frac{2\log(2n^2d^2)}{C_0^2} (\log s_0)(\log \log nd) \rfloor, \\ \frac{d}{16\lambda_0^2\log(nd)C_1^2\gamma^2}\log(\frac{d}{\lambda_0^2\log(nd)C_1^2\gamma^2})\Big)$ with $C_0 = \min \{\frac{1}{2}, \frac{\phi_0^2}{512s_0s_A^2\nu\omega_{\mathcal{X}}}\}$, $C_1 = \frac{4\nu\omega_{\mathcal{X}}s_0}{\phi_0^2}+ \sqrt{(1+\frac{4\nu\omega_{\mathcal{X}}}{\phi_0^2})s_0}$. The following theorem provides a regret upper bound for NELA algorithm. 

%-------------------------------------->Theorem
\begin{theorem}
\label{the:rgt}
Under Assumptions~\ref{assump:arfv}-\ref{assump:spd}, for all $T \geq \tau$, with probability at least $1-\delta$, the regret of NELA algorithm satisfies
\begin{align}
    R(T) &= O\Big(dn\log(1+\frac{\sum_{t=1}^T\sum_{j=1}^n \boldsymbol{W}(u_t,j)^2}{dn \lambda_1})\sqrt{T}\Big) \notag \\
    &\quad + O(s_As_1\tau + \log s_0\sqrt{s_0T}). \label{eq:rgt}
\end{align}
\end{theorem}
%---------------------------
%------------------------
\begin{proof}
    See Appendix~\ref{app:the}. 
\end{proof}
%------------------------------
%-------------------------->Remark
\begin{remark}
The regret bound consists of three main terms: $O\Big(dn\log(1+\frac{\sum_{t=1}^T\sum_{j=1}^n \boldsymbol{W}(u_t,j)^2}{dn \lambda_1})\sqrt{T}\Big)$, $O(s_As_1\tau)$ and $O(\log s_0\sqrt{s_0T})$. The first term $O\Big(dn\log(1+\frac{\sum_{t=1}^T\sum_{j=1}^n \boldsymbol{W}(u_t,j)^2}{dn \lambda_1})\sqrt{T}\Big)$ captures how the graph structure influences regret. When users are isolated (i.e., no influence, $\boldsymbol{W}$ is an identity matrix), this term reduces to the standard linear bandit regret, $O(dn\log(1+\frac{T}{dn\lambda_1})\sqrt{T})$. No collaborative benefit is possible in this case. When the graph is fully connected and uniform (i.e., each user equally influences all others, $\boldsymbol{W}(i,j) = 1/n, \forall i, j \in \mathcal{U}$), the regret can improve to $O(dn\log(1+\frac{T}{n^2d\lambda_1})\sqrt{T})$ due to increased collaborative information. In general, collaboration helps when the influence weights are distributed broadly across many neighbors, i.e. low concentration in $\sum_{j=1}^n\boldsymbol{W}(u_t, j)^2$, $\forall u_t \in \mathcal{U}$, as this increases the effective sample size for each user while highly skewed or low-connectivity patterns yield smaller gains.

The second term $O(s_As_1\tau)$ characterizes the difficulty of identifying sparse anomalies. Larger $\gamma$ makes detection easier (smaller $\tau$), reducing this term. The term $O(\log s_0 \sqrt{s_0T})$ accounts for the influence of anomalous deviations, with $s_0$ being proportional to the number of anomalies and $s_0 \ll nd$. In the absence of anomalies, our setting degenerates to the classical collaborative networked bandit setting, as described in \cite{wu2016contextual}. 
\end{remark}
%---------------------------
\begin{remark}
The regret bound in \eqref{eq:rgt} does not explicitly depend on the positions of anomalies in the network. This stems from the model formulation, where each user $i$'s feature vector is decomposed as $\boldsymbol{\theta}_i + \boldsymbol{v}_i$, with $\boldsymbol{\theta}_i$ governing collaborative behavior under the graph structure defined by  $\boldsymbol{W}$ and  $\boldsymbol{v}_i$ representing sparse anomalous deviations. As a result, the mutual influence and learning benefit across the graph are determined solely by the clean feature matrix $\boldsymbol{\Theta}$, not by the locations of anomalous nodes. 

In the regret analysis, the contribution of anomalies is captured through global sparsity parameters such as $s_0$ and the incorrect detection is controlled within a bounded time window $\tau$. Any potential effect of anomaly position, if it exists, is therefore transient and bounded in the early phase of learning, and does not dominate the overall regret.
\end{remark}
%-------------------------
%--------------------------
\begin{theorem}\label{the:lb}
Under Assumptions~\ref{assump:arfv}-\ref{assump:spd}, for all $T \geq 2$, and any learning algorithm $\pi$, there exists a networked linear bandit instance such that the regret is lower bounded by
\begin{align}
    R(T) \geq  \Omega\Big(\xi d\sqrt{nT}+\sqrt{s_0T}\Big),
\end{align}
where $\xi \in [1, \sqrt{n}]$ is a problem-dependent scalar that captures how much collaborative information can be shared across users. If all users share the same parameter vector $\boldsymbol{\theta}$ and the matrix $\boldsymbol{W}$ is fully mixing (each user aggregates information uniformly from others), then $\xi = 1$, while if users are completely isolated (no information sharing, $\boldsymbol{W} = \boldsymbol{I}$ ) and each user has an independent $\boldsymbol{\theta}$, then $\xi = \sqrt{n}$. 
\end{theorem}
%--------------------------------
\begin{proof}
See Appendix~\ref{app:the-lb}.
\end{proof}
%--------------------------------
\begin{remark}
The computational scalability of NELA is influenced by the number of users $n$, the dimension $d$, as the matrix $\boldsymbol{A}_t$ in Algorithm 1 is of size $nd \times nd$, potentially limiting the scalability of NELA in very large systems. However, this behavior is consistent with that of existing graph-based contextual bandit algorithms \cite{cesa2013gang,wu2016contextual}. Besides, NELA leverages the sparsity in the graph structure (i.e., sparse $\boldsymbol{W}$) and the anomaly patterns (sparse $\boldsymbol{v}_i$), which helps reduce the effective dimensionality and computational burden in practice.
\end{remark}

%--------------------->Experiments
\section{Experiments}\label{sec:experiment}
We test the performance of our proposed algorithm on both synthetic and real-world datasets. Besides, we compare our algorithm with the following state-of-the-art algorithms:
%----------------------------->SOTA
\begin{itemize} 
    \item CoLin \cite{wu2016contextual}: Each user shares their instantaneous payoff with neighbors in the graph. The users determine the payoff jointly by using a collaborative bandit algorithm;
    \item GraphUCB \cite{yang2020laplacian}: Each user shares their instantaneous payoff with neighbors in the graph and the feature vectors of users are smoothed along the graph;  
    \item nLinUCB: The algorithm maintains $n$ independent LinUCB algorithms \cite{abbasi2011improved} for each user, where the bandit uses no graph information;
    \item LinUCB: The algorithm maintains a single LinUCB algorithm for $n$ users globally, equivalent to running GraphUCB with an identity Laplacian matrix or the CoLin algorithm with an identity weight matrix.
\end{itemize}
%---------------------------
All the results are based on the average of ten independent runs. In all experiments, we set the confidence probability parameter $\delta = 0.001$, noise variance $\sigma = 0.01$, regularized parameter $\lambda_0 = 0.02$ and $\lambda_1 = 1$. We observed that the algorithm is not overly sensitive to moderate changes in these hyperparameters. Therefore, we select hyperparameters in consistent with \cite{yang2020laplacian,ariu2022thresholded}. To ensure a controlled and fair comparison, all experimental settings are held constant across different runs within the same experiment.

%----------------------------->
\subsection{Experiment on Toy Examples}
We first present experiments on toy examples to better explain our algorithm. Specifically, we implemt with two types of graph structures: undirected star graphs and fully connected graphs (Figure~\ref{fig:toy-example}), aiming to validate the effectiveness of collaboration among agents under different graph structures. We assume that the users are homogeneous, with each having an equal influence on its neighbors ($\boldsymbol{W}(i,j) = 1/(1+\sum_{j=1}^n\mathbbm{1}(e(i,j) \in \mathcal{E}))$) and node 1 is designated as an anomaly. At each round $t$, we sample arm set $\mathcal{A}_{u_t,t}$ from a multivariate Gaussian distribution $\mathcal{N}(\boldsymbol{0}_{M}, V)$, where $M = 50$, $V_{i,i} = 1$ for all $i \in [M]$, and $V_{i,m} = \rho^2 = 0.7$ for all $i \neq k$. We then normalize the context vectors so that $\norm{\boldsymbol{x}} \leq 1$ for all $\boldsymbol{x} \in \mathcal{A}_{u_t,t}$. We test various levels of $\gamma$ representing the degree to which the anomaly deviates from the majority. The results are presented in Figure~\ref{fig:comp-toy}. 

%-------------------------------------->Figure
\begin{figure}[!ht]
\centering
\begin{subfigure}{.2\textwidth}
 \centering
  % include first image
  \includegraphics[width=.9\linewidth]{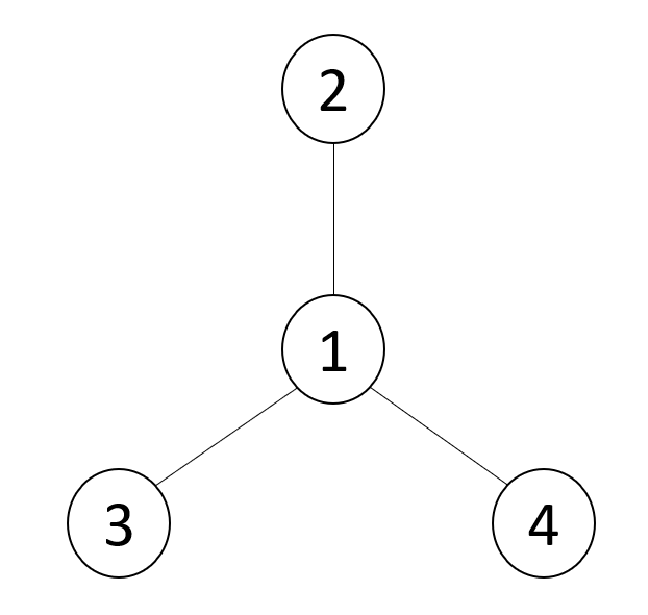}
  \caption{Star graph.}
\end{subfigure}
\begin{subfigure}{.2\textwidth}
 \centering
  % include first image
  \includegraphics[width=.9\linewidth]{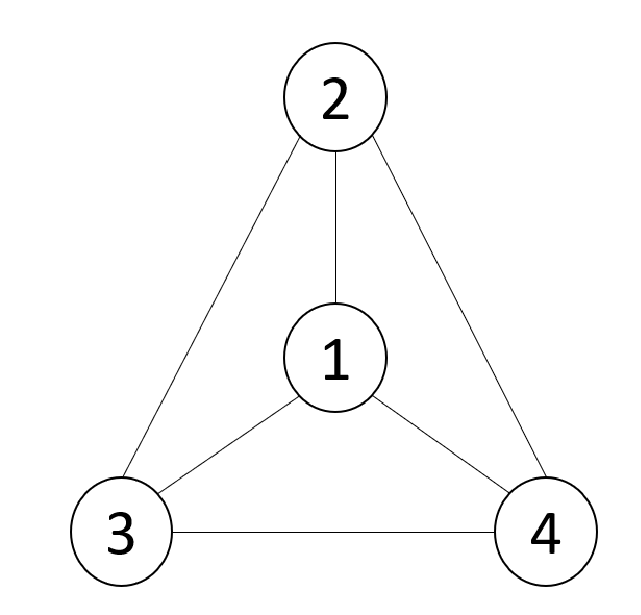}
  \caption{Fully connected graph.}
\end{subfigure}
\caption{Toy Example on Four Nodes}
\label{fig:toy-example}
\end{figure}
%------------------------

%---------------->Figure
\begin{figure}[H]
    \centering
    \includegraphics[width=0.8\linewidth]{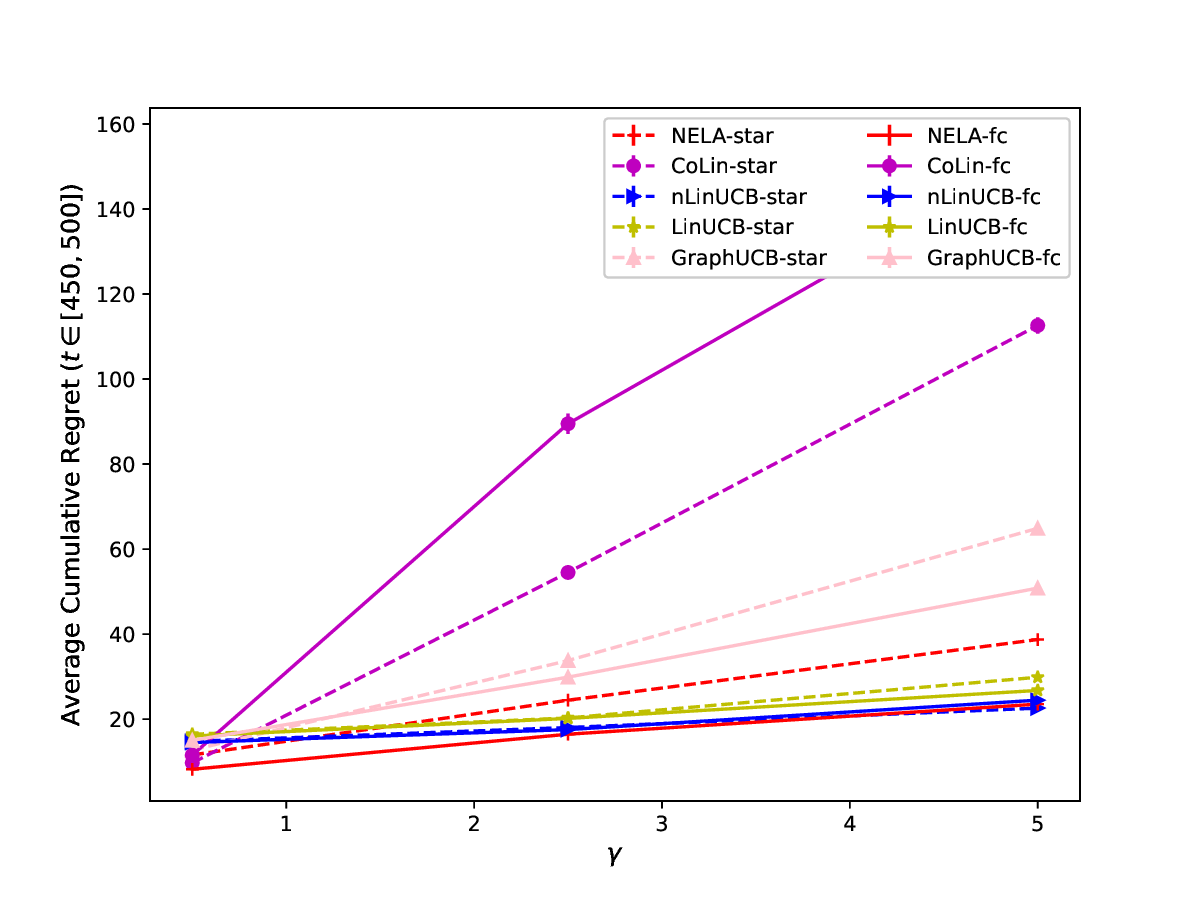}
    \caption{Cumulative regret under different settings.}
    \label{fig:comp-toy}
\end{figure}
%------------------

In Figure~\ref{fig:comp-toy}, plots with dashed lines are results obtained from experiments on a star graph, while plots with solid lines are results obtained from experiments on a fully connected graph. According to Figure~\ref{fig:comp-toy}, NELA and GraphUCB demonstrate better performance on fully connected graphs compared to star graphs. This improvement arises because of effective use of the graph structure, which enables greater collaboration among agents as the number of connections increases. In contrast, the CoLin algorithm, although it also incorporates graph information, performs worse on fully connected graphs. This is because CoLin assumes fully cooperative scenarios, making it more vulnerable to performance degradation in the presence of anomalies. Meanwhile, the performance of nLinUCB and LinUCB shows little variation between the two graph structures and under varying $\gamma$, as these algorithms treat each node independently, disregarding the network structure.

When $\gamma$ is very small ($\gamma < 0.5$), all algorithms perform similarly, with NELA and CoLin achieving the best results. At these low values of $\gamma$, the influence of the anomaly is minimal and can be even ignored, allowing collaboration among agents to significantly enhance performance. However, as $\gamma$ increases, regrets also rise, with the most significant deterioration observed in the CoLin algorithm. At larger $\gamma$, the perturbations caused by anomalies become more pronounced, disrupting the ideal assumption that all nodes adhere to the same behavior pattern. This mismatch exacerbates CoLin’s reliance on collaboration, magnifying the detrimental impact of anomalies on performance.

Under the star graph, LinUCB and nLinUCB algorithms outperform NELA when $\gamma > 1$. While NELA leverages graph information to enhance performance, its advantages are less evident in extremely simple graph structures. In such cases, the additional computations required for utilizing the graph information may reduce efficiency, resulting in lower performance compared to algorithms that treat nodes independently. However, as the network size increases, even in sparsely connected scenarios, as demonstrated in experiments on synthetic and real-world datasets where NELA's advantages become more evident.
%---------------------------

%----------------------------------Experiments on synthetic dataset
\subsection{Experiments on synthetic dataset}
We generate a synthetic dataset with $n=50$ users, each assigned a $d=10$ dimensional feature vector $\boldsymbol{\theta}_i$. Each dimension of $\boldsymbol{\theta}_i$ is drawn from a normal distribution $\mathcal{N}(0,1)$ and then normalized so that $\norm{\boldsymbol{\theta}_i}=1$. We then construct the weight matrix $\boldsymbol{W}$ for the user graph by setting $\boldsymbol{W}(i,j) \propto \langle \boldsymbol{\theta}_i, \boldsymbol{\theta}_j \rangle$ and normalizing each column of $\boldsymbol{W}$ by its $\ell_1$ norm. To control the connectivity of each node, we introduce a threshold that filters out the weakest 70\% of the weights in the network, ensuring that only the top 30\% of weights are retained. The generation of the arm set follows the same procedure as in the toy example experiments. We randomly sample three anomalies from $n =50$ users. For each anomaly, one dimension of its residual vector is sampled from the uniform distribution $U(-10,10)$, resulting in sampled residual norms of 2.095, 3.33, and 5.44, respectively. The norm of the residual feature vectors for the non-anomalous users is zero, i.e., $\norm{\boldsymbol{v}_i} = \boldsymbol{0}$, $i \in \mathcal{U}\backslash \tilde{\mathcal{U}}$. 
%-------------------------------------->Figure
\begin{figure*}[ht]
\centering
\begin{subfigure}{.35\textwidth}
 \centering
  % include first image
  \includegraphics[width=.99\linewidth]{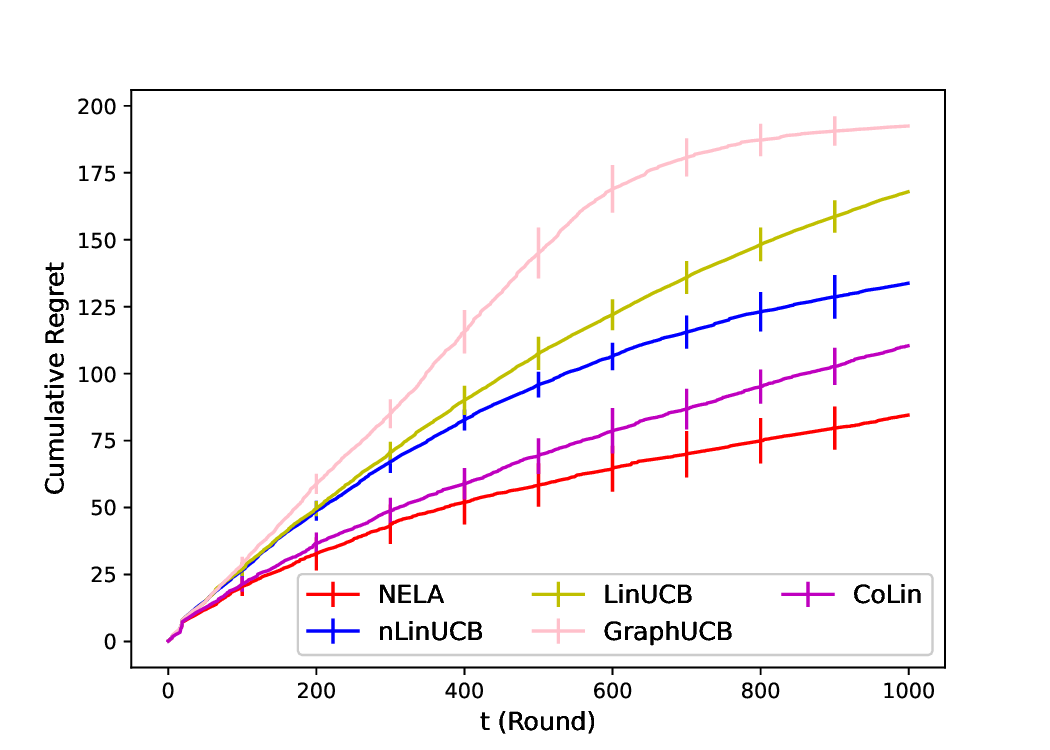}
  \caption{Cumulative Regret}
  \label{fig:e1-rgt}
\end{subfigure}
\begin{subfigure}{.3\textwidth}
 \centering
  % include first image
  \includegraphics[width=.99\linewidth]{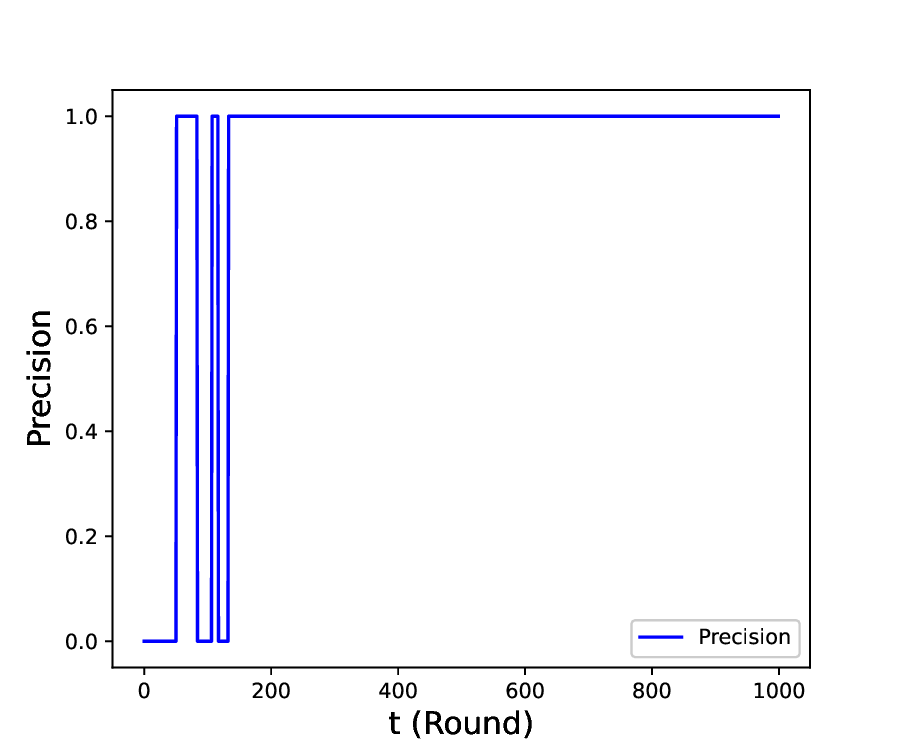}
  \caption{Precision of per-round detection}
  \label{fig:e1-p}
\end{subfigure}
\begin{subfigure}{.3\textwidth}
  \centering
  % include second image
  \includegraphics[width=.99\linewidth]{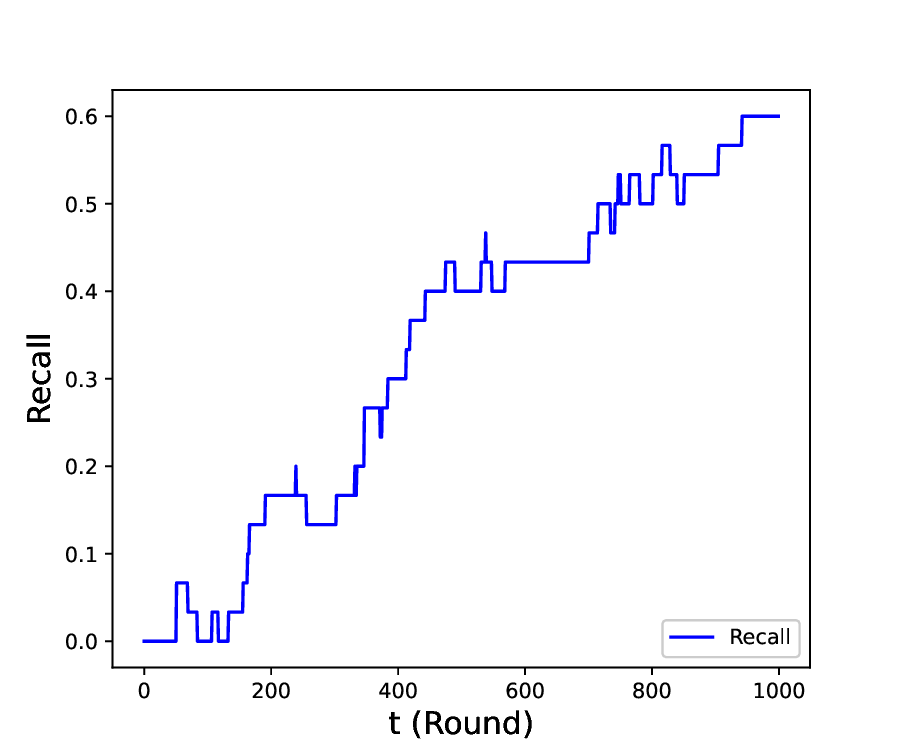} 
  \caption{Recall of per-round detection}
  \label{fig:e1-r}
\end{subfigure}
\caption{Analysis of regret and anomaly detection (Synthetic dataset with $n =50$).}
\label{fig:e1-n50}
\end{figure*}
%---------------------------------------
%-------------------------------------->Figure
\begin{figure*}
\centering
\begin{subfigure}{.35\textwidth}
 \centering
  % include first image
  \includegraphics[width=.99\linewidth]{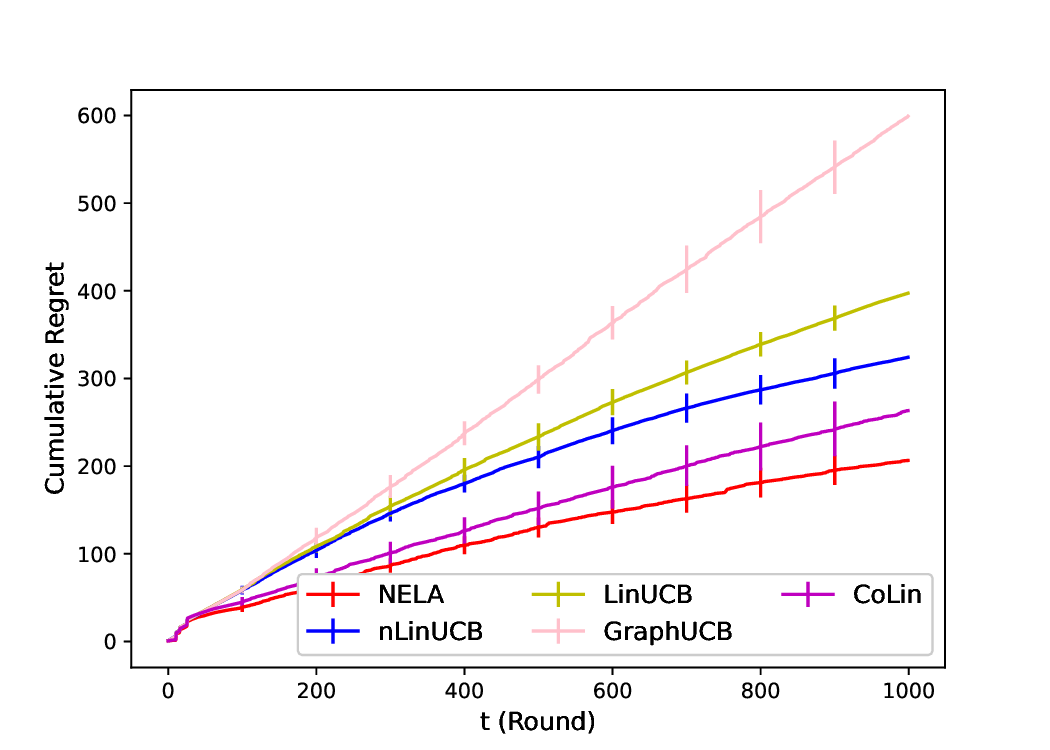}
  \caption{Cumulative Regret.}
\end{subfigure}
\begin{subfigure}{.3\textwidth}
 \centering
  % include first image
  \includegraphics[width=.99\linewidth]{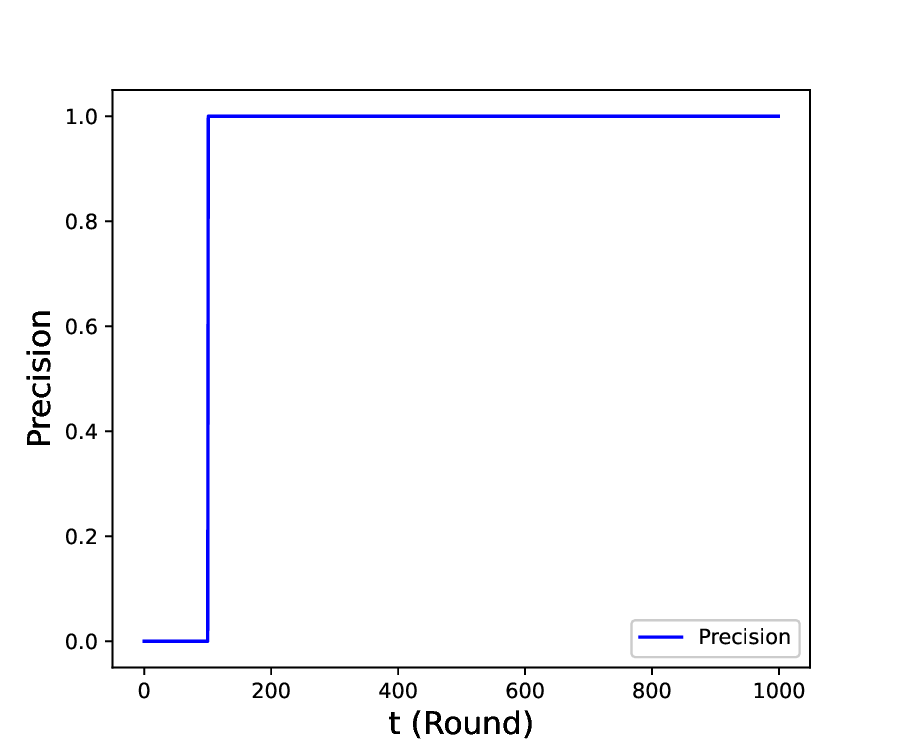}
  \caption{Precision of per-round detection.}
  \label{fig:e2-p}
\end{subfigure}
\begin{subfigure}{.3\textwidth}
  \centering
  % include second image
  \includegraphics[width=.99\linewidth]{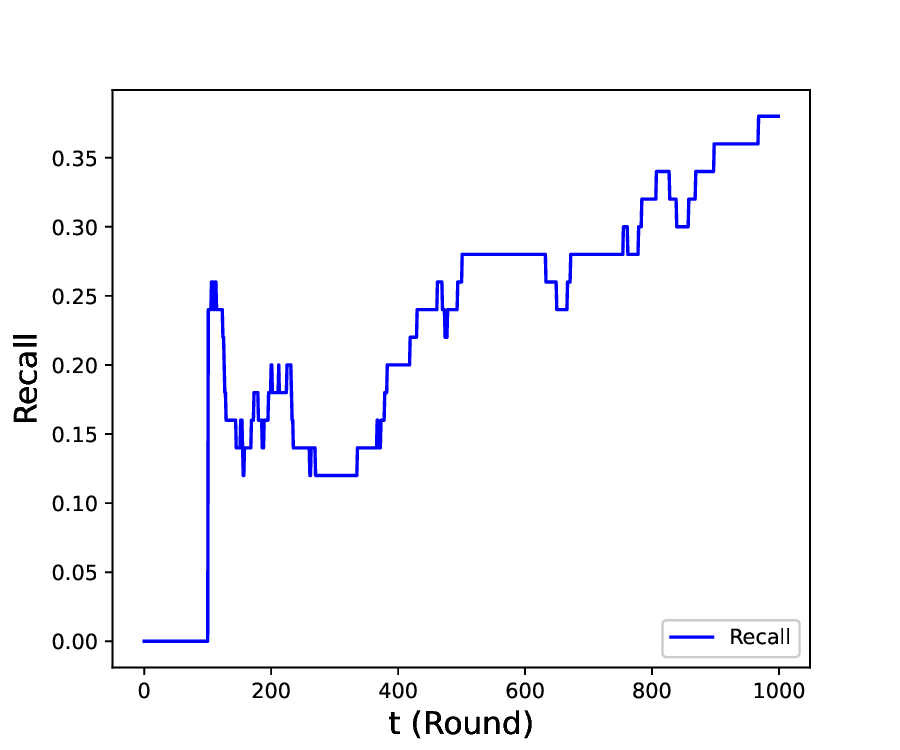} 
  \caption{Recall of per-round detection.}
  \label{fig:e2-r}
\end{subfigure}
\caption{Analysis of regret and anomaly detection (Real-world dataset with $n =100$).}
\label{fig:e2-n100}
\end{figure*}
%------------------------

Figure~\ref{fig:e1-n50} demonstrates the performance of our proposed algorithm NELA in terms of regret and anomaly detection. According to Figure~\ref{fig:e1-rgt}, NELA outperforms all state-of-the-art algorithms and has the lowest regret. GraphUCB, which assumes the smoothness of feature vectors along the network, is highly susceptible to anomalies. This explains its highest regret among the compared algorithms. Algorithms nLinUCB and LinUCB infer users' feature vectors without considering mutual influence and information sharing among users, leading to higher regrets than CoLin because of underscoring the significance of incorporating users' mutual influence in the algorithm design. Our proposed algorithm, NELA, not only considers these effects but also addresses the impact of potential anomalies, thus achieving the lowest regret. 

Figures~\ref{fig:e1-p} and \ref{fig:e1-r} show the anomaly detection results. As stated in section~\ref{sec:problem}, we evaluate anomaly detection performance using two key metrics: i) precision, which measures the accuracy of the detected anomalies by quantifying the proportion of true anomalies among all detections, and ii) recall, which measures the completeness by assessing the proportion of true anomalies that are successfully detected. In Figure~\ref{fig:e1-p}, the precision quickly converges to one, indicating that the NELA algorithm rarely misclassifies normal users as anomalies. Figure~\ref{fig:e1-r} depicts the proportion of detected anomalies with respect to the ground truth. Our experiment corresponds to a particularly challenging scenario, where each anomaly has only one non-zero dimension. In Figure~\ref{fig:e1-n50-level}, we provide detection results under varying $\gamma$ values and different numbers of non-zero dimensions $|\mathcal{J}(\boldsymbol{v}_i)|$, $i \in\tilde{\mathcal{U}}$, reflecting various levels of detection difficulties. Since NELA guarantees zero false detection empirically, we plot the recalls in Figure~\ref{fig:e1-n50-level}.
%---------------->Figure
\begin{figure}
    \centering
    \includegraphics[width=0.99\linewidth]{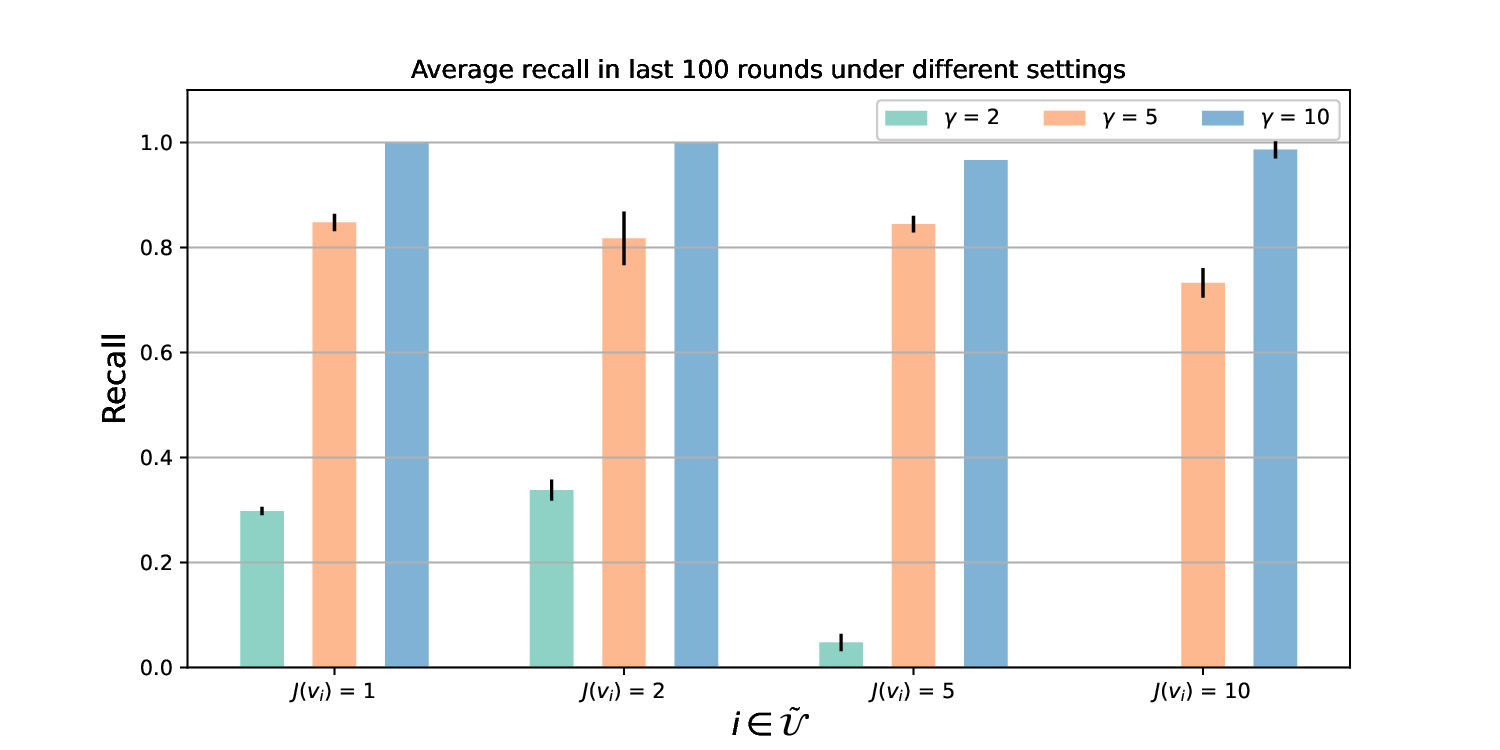}
    \caption{Anomaly detection under different settings.}
    \label{fig:e1-n50-level}
\end{figure}
%------------------
According to Figure~\ref{fig:e1-n50-level}, the anomaly detection performance improves with increasing $\gamma$. That aligns with the intuition that a bigger $\gamma$ makes anomalies more detectable. However, there is a trade-off in the number of non-zero dimensions. Since NELA detects anomalies by identifying the support of the global residual vector $vec(\boldsymbol{V})$, anomalies with more non-zero dimensions are simpler to detect. Detecting even one of the non-zero dimensions in one anomaly is sufficient to identify it; nevertheless, with a fixed $\gamma$, having more non-zero dimensions results in a smaller value per dimension, making detection harder. That explains the noticeable decrease in recalls concerning $\mathcal{J}(\boldsymbol{v}_i)$ for $i \in\tilde{\mathcal{U}}$ with $\gamma = 2$ while such differences in experiments with $\gamma = 5$ and $\gamma = 10$ are not significant.
%----------------------------------Experiments on real dataset
\subsection{Experiments on real-world dataset}
Netflix Dataset \footnote{Netflix Movie Rating Dataset from Netflix's ``Netflix Prize" competition on \url{https://www.kaggle.com/datasets/rishitjavia/netflix-movie-rating-dataset?resource=download}} \cite{bennett2007netflix} contains the information of more than 16 million ratings for 1350 movies by 143458 users, which has been widely used to evaluate the effectiveness of contextual bandit algorithms \cite{yang2020laplacian,cheng2023parallel,silva2023exploring}. For our experiments, we select $10^3$ movies with the most ratings and $n = 100$ most active users to construct a sparse user–item collaborative filtering matrix, where unrated pairs are assigned a value of zero. We then apply singular value decomposition (SVD) to this matrix to compute  (dimension $d = 10$) user feature vectors and movie (context) vectors, following the methodology in \cite{li2019improved, cheng2023parallel}. Before using the feature and context vectors in our experiments, we normalize the data in two steps. First, we apply min-max scaling to each dimension to map the values into the range $[-1,1]$. Then, we perform $\ell_2$-normalization on each vector to ensure unit norm.

One challenge is the absence of information about an explicit user connectivity graph or labeled anomalies, which limits the representation of realistic social network scenarios. To bridge this gap, we follow established methods in prior work \cite{wu2016contextual,yang2020laplacian} to construct a synthetic graph based on feature proximity, treating it as ground truth. 

At each step, we sample $M = 50$ items from the extracted $10^3$ movies as the action (context) set. The weight matrix construction are the same as the synthetic dataset. While the user and item (context) features are extracted from the real dataset, the reward is generated according to our payoff model in \eqref{eq:payoff}. This setting allows us to simulate bandit feedback realistically while ensuring that the environment is consistent with our theoretical model and supports evaluation of both regret minimization and anomaly detection. We sample five anomalies following the same setting as the synthetic dataset, and the sampled norms of residuals are 8.46, 1.40, 6.54, 1.26, and 8.43, respectively. 

Figure~\ref{fig:e2-n100} shows the results. Similar to the experiment on the synthetic dataset, NELA exhibits superior performance over all state-of-the-art algorithms, demonstrating the effectiveness of our proposed algorithm. However, the recall of anomaly detection in Figure~\ref{fig:e2-r} is lower in the real-world dataset compared to the synthetic dataset. The reason is that in the former, the arm set represents the features of available movies. Consequently, Assumption~\ref{assump:context-d} is not guaranteed to be satisfied, which affects the algorithm's performance. On the other hand, the precision plot of NELA in Figure~\ref{fig:e2-p} shows that it guarantees zero false detection performance as the experiments on the synthetic dataset. 

%---------------------------->Conclusion
\section{Conclusion}\label{sec:conclusion}
We studied the joint action recommendation and anomaly detection problem within an online learning networked bandits setting. Our proposed algorithm, NELA, improves the state-of-the-art by integrating one module for anomaly detection to ensure robust performance. Theoretical analysis and empirical experiments on synthetic and real-world datasets validate the effectiveness of NELA in anomaly detection and demonstrate its superiority over state-of-the-art algorithms. In future work, we plan to explore more realistic and dynamic environments. One key direction is to extend our framework to settings where the user network is unknown or evolves over time, which better reflects practical systems. Additionally, although our theoretical regret bounds do not depend on the positions of anomalies, empirical performance may still be influenced by where anomalies occur, particularly if they are situated at highly connected nodes. Understanding such effects under imperfect anomaly detection is a promising line of study. Besides, certain coordinated attacks or bot networks may alter the influence structure. Therefore, extending the framework to jointly model changes in both the influence structure and the residual deviations is also an interesting direction for future work.

\appendix
\subsection{Auxiliary Lemmas}
\label{app:aux}
In this section, we provide some auxiliary theorems and lemmas that are necessary for the proof of the regret bound. 
%We begin by presenting necessary theorem from the literature for proofs. Then we present our auxiliary results in the form of lemmas. These results offer performance guarantees for the Lasso estimator. Specifically, we adapt the related results from \cite{oh2021sparsity,ariu2022thresholded} to the alternating method and multi-user setting in our problem formulation.  
%-------------------------->Theorem
\begin{theorem}[Matrix Chernhoff, Theorem 3.1 of \cite{tropp2011user}]\label{the:mc}
Let $\mathcal{F}_1 \subset \mathcal{F}_2 \subset \ldots \subset \mathcal{F}_t$ be a filtration and consider a finite sequence $\{X_s\}$ of positive semi-definite matrix with dimensions $d$, adapted to the filtration. Suppose $\lambda_{\max}(X_k) \leq R$ almost surely. Define finite series $Y = \sum_{s=1}^t X_s$ and $W = \sum_{s=1}^t\mathbb{E}[X_s|\mathcal{F}_{s-1}]$. Then, for all $\mu \geq 0$, for all $\delta \in [0,1)$, we have
\begin{align}
    \mathbb{P}(\lambda_{\min}(Y) \leq (1-\delta)\mu \text{ and } &\lambda_{\min}(W) \geq \mu) \notag \\
    &\leq d(\frac{e^{-\delta}}{(1-\delta)^{1-\delta}})^{\frac{\mu}{R}}.
\end{align}
\end{theorem}
%----------------------

%---------------------------->Corollary
\begin{corollary}[Corollary 6.8 in \cite{buhlmann2011statistics}]\label{cor:68}
Suppose $\boldsymbol{\Sigma}_0$ satisfies the compatibility condition for the set $S$ with $|S| = s_0$ and with the compatibility constant $\phi^2(\boldsymbol{\Sigma}_0, S) > 0$ and that $\norm{\boldsymbol{\Sigma}_0 - \boldsymbol{\Sigma}_1} \leq \lambda$ where $\frac{32\lambda s_0}{\phi^2(\boldsymbol{\Sigma}_0, S)} \leq 1$. Then the compatibility condition also holds for $\boldsymbol{\Sigma}_1$ with the compatibility constant $\frac{\phi^2(\boldsymbol{\Sigma}_0, S)}{2}$, i.e. $\phi^2(\boldsymbol{\Sigma}_1, S) \geq \frac{\phi^2(\boldsymbol{\Sigma}_0, S)}{2}$. 
\end{corollary}
%---------------------------

%-------------------------------->Theorem
\begin{theorem}[Theorem 2.19 in \cite{wainwright2019high}]\label{the:219}
Let $(Z_t, \tilde{\mathcal{F}}_t)_{t = 1}^{\infty}$ be a martingale difference sequence, and assume that for all $\alpha \in \mathbb{R}$, $\mathbb{E}[\exp(\alpha Z_s)|\tilde{\mathcal{F}}_{s-1}] \leq \exp(\frac{\alpha^2 \sigma^2}{2})$. Then for all $x \geq 0$, we have
\begin{gather*}
    \mathbb{P}(|\sum_{s=1}^t Z_s| \geq x) \leq 2\exp(-\frac{x^2}{2t\sigma^2}).
\end{gather*}
\end{theorem}
%----------------------------------------

%------------------------>Lemma
\begin{lemma}[Bernstein-like inequality for the adapted data (Lemma 8 in \cite{oh2021sparsity})]\label{lem:bl-ineq}
Suppose for all $t \geq 1$, for all $1 \leq i \leq j \leq nd$, $\mathbb{E}[\gamma_t^{ij}(vec(\mathring{\boldsymbol{X}}_{a_s,u_s}))|\mathcal{F}_{t-1}] = 0$ and $\mathbb{E}[|\gamma_t^{ij}(vec(\mathring{\boldsymbol{X}}_{a_s,u_s}))|^m|\mathcal{F}_{t-1}] \leq m!$ for all integer $m \geq 2$. Then for all $x > 0$ and for all integer $t \geq 1$, we have:
\begin{align*}
  \mathbb{P}\Big (& \max_{1 \leq i \leq j \leq d} \Big |\frac{1}{t}\sum_{s=1}^t \gamma_t^{ij}(vec(\mathring{\boldsymbol{X}}_{a_s,u_s})) \Big| \geq x + \sqrt{2x} \notag \\
  & + \sqrt{\frac{4\log (2n^2d^2)}{t}} + \frac{2\log(2n^2d^2)}{t} \Big) \leq \exp(-\frac{tx}{2}). 
\end{align*}
\end{lemma}
%----------------------------

Define the loss function
\begin{align*}
l_t(vec(\boldsymbol{V})) = & \frac{1}{t}\sum_{s=1}^t (r_{a_s,u_s} - vec(\mathring{\boldsymbol{X}}_{a_s,u_s})^{\top}vec(\widehat{\boldsymbol{\Theta}}_s\boldsymbol{W}) \notag \\
&- vec(\mathring{\boldsymbol{X}}_{a_s,u_s})^{\top}vec(\boldsymbol{V}))^2. \label{eq:loss}
\end{align*}
The initial LASSO estimate is given by
\begin{gather*}
   \widehat{\boldsymbol{V}}_0^t = \arg \min_{\boldsymbol{V}}\{l_t(vec(\boldsymbol{V})) + \lambda_t\norm{vec(\boldsymbol{V})}_1\}.
\end{gather*}
Define $\boldsymbol{\Delta} = \widehat{\boldsymbol{V}}_0^t - vec(\boldsymbol{V})$. We analyze the performance of the initial Lasso estimate. 
%------------------------------>Lemma
\begin{lemma}
\label{lem:ehat-v}
Denoting $\mathcal{E}(vec(\boldsymbol{V}')) = \mathbb{E}[l_t(vec(\boldsymbol{V}'))] - \mathbb{E}[l_t(vec(\boldsymbol{V}))]$, it is satisfied that
\begin{gather*}
    \mathcal{E}(\widehat{\boldsymbol{V}}_0^t) \geq \frac{1}{2}(\widehat{\boldsymbol{V}}_0^t-vec(\boldsymbol{V}))^{\top}\widehat{\overline{\boldsymbol{\Sigma}}}_t(\widehat{\boldsymbol{V}}_0^t-vec(\boldsymbol{V}))
\end{gather*}
\end{lemma}
%--------------------------------
%------------------------>Proof
\begin{proof}
Denote $\widehat{\boldsymbol{V}} = \widehat{\boldsymbol{V}}_0^t$ for simplicity. From the definition of $\mathcal{E}(vec(\boldsymbol{V}'))$, we have
\begin{align}
\mathcal{E}(\widehat{\boldsymbol{V}}_0^t) &= \mathbb{E}[l_t(\widehat{\boldsymbol{V}}_0^t)] - \mathbb{E}[l_t(vec(\boldsymbol{V}))] \notag \\
%&= \frac{1}{t}\mathbb{E}[\sum_{s=1}^t ((vec(\mathring{\boldsymbol{X}}_{a_s,u_s})^T[vec(\boldsymbol{\Theta}\boldsymbol{W}) + vec(\boldsymbol{V})] + \eta_s - vec(\mathring{\boldsymbol{X}}_{a_s,u_s})^T[vec(\widehat{\boldsymbol{\Theta}}_s\boldsymbol{W}) + \widehat{\boldsymbol{V}}_0^s])^2 \notag \\
%& \quad - \sum_{s=1}^t (vec(\mathring{\boldsymbol{X}}_{a_s,u_s})^T(vec(\boldsymbol{\Theta}\boldsymbol{W})- vec(\boldsymbol{\Theta}\boldsymbol{W}))+\eta_s)^2] \notag \\
%&= \frac{1}{t}\sum_{s=1}^t(\widehat{\boldsymbol{V}}_0^s - vec(\boldsymbol{V}))^T vec(\mathring{\boldsymbol{X}}_{a_s,u_s})vec(\mathring{\boldsymbol{X}}_{a_s,u_s})^T(\widehat{\boldsymbol{V}}_0^s - vec(\boldsymbol{V})) \notag \\
%&\quad + \frac{1}{t}\sum_{s=1}^t(vec(\widehat{\boldsymbol{\Theta}}_s\boldsymbol{W})-vec(\boldsymbol{\Theta}\boldsymbol{W}))^Tvec(\mathring{\boldsymbol{X}}_{a_s,u_s})vec(\mathring{\boldsymbol{X}}_{a_s,u_s})^T(vec(\widehat{\boldsymbol{\Theta}}_s\boldsymbol{W})-vec(\boldsymbol{\Theta}\boldsymbol{W})) \notag \\
&= (\widehat{\boldsymbol{V}} - vec(\boldsymbol{V}))^{\top}\widehat{\overline{\boldsymbol{\Sigma}}}_t(\widehat{\boldsymbol{V}} - vec(\boldsymbol{V})) +(vec(\widehat{\boldsymbol{\Theta}}_t\boldsymbol{W}) \notag \\
& \quad - vec(\boldsymbol{\Theta}\boldsymbol{W}))^{\top}\widehat{\overline{\boldsymbol{\Sigma}}}_t(vec(\widehat{\boldsymbol{\Theta}}_t\boldsymbol{W})-vec(\boldsymbol{\Theta}\boldsymbol{W})) \notag \\
&\geq \frac{1}{2}(\widehat{\boldsymbol{V}} - vec(\boldsymbol{V}))^{\top}\widehat{\overline{\boldsymbol{\Sigma}}}_t(\widehat{\boldsymbol{V}} - vec(\boldsymbol{V})),
\end{align}
where for the inequalities we used the positive semi-definiteness of $\widehat{\overline{\boldsymbol{\Sigma}}}_t$.   
\end{proof}
%-----------------------------------------
%------------------------>Lemma
\begin{lemma}
\label{lem:epsilon-v-p}
Denoting $\varepsilon_t(vec(\boldsymbol{V})) = l_t(vec(\boldsymbol{V})) -  \mathbb{E}[l_t(vec(\boldsymbol{V}))]$, we have: 
\begin{align}
    \mathbb{P} &\Big( | \varepsilon_t(\widehat{\boldsymbol{V}}_0^t)) - \varepsilon_t(vec(\boldsymbol{V}))| \leq \frac{1}{2}\lambda_t \norm{\widehat{\boldsymbol{V}}_0^t - vec(\boldsymbol{V})}_1 \Big) \notag \\
    &\geq 1 - 2\exp\Big(-\frac{t\lambda_t^2}{32\sigma^2s_A^2} + \log nd\Big)
\end{align}
\end{lemma}
%-------------------------
%
\begin{proof}
Denote $\widehat{\boldsymbol{V}} = \widehat{\boldsymbol{V}}_0^t$. We compute $\varepsilon_t(\widehat{\boldsymbol{V}})$ as
\begin{align}
    \varepsilon_t(\widehat{\boldsymbol{V}}) &= l_t(\boldsymbol{V}) -  \mathbb{E}[l_t(vec(\boldsymbol{V}))] \notag 
\end{align}
Besides,
\begin{align}
    \varepsilon_t(vec(\boldsymbol{V})) &= \frac{1}{t}\sum_{s=1}^t \Big( 2\eta_s vec(\mathring{\boldsymbol{X}}_{a_s,u_s})^{\top} [vec(\boldsymbol{\Theta}\boldsymbol{W})  \notag \\
    &\quad - vec(\widehat{\boldsymbol{\Theta}}_s\boldsymbol{W})] + \eta_s^2 - \mathbb{E}[\eta_s^2] \Big). 
\end{align}
Thus, we can compute
\begin{align*}
    \varepsilon_t(\widehat{\boldsymbol{V}}) - \varepsilon_t&(vec(\boldsymbol{V})) = \frac{1}{t}\sum_{s=1}^t 2\eta_s vec(\mathring{\boldsymbol{X}}_{a_t,u_t})^{\top} [vec(\boldsymbol{V})-\widehat{\boldsymbol{V}}] \notag \\
    &\leq \frac{2}{t}\norm{\sum_{s=1}^t\eta_svec(\mathring{\boldsymbol{X}}_{a_s,u_s})}_{\infty}\norm{vec(\boldsymbol{V})-\widehat{\boldsymbol{V}}}_1.
\end{align*}
Using H\"oder's inequality we have
\begin{align*}
    \mathbb{P} &\Big(\frac{2}{t}\norm{\sum_{s=1}^t\eta_svec(\mathring{\boldsymbol{X}}_{a_s,u_s})}_{\infty} \leq \lambda \Big) \notag \\
    &\geq 1 - \sum_{i=1}^{nd}\mathbb{P}\Big (\frac{2}{t}|\sum_{s=1}^t\eta_svec(\mathring{\boldsymbol{X}}_{a_s,u_s})(i)| > \lambda \Big),
\end{align*}
where $\cdot(i)$ is an indicator of the $i-$th element. Similar to $\boldsymbol{S}_t$, $\eta_svec(\mathring{\boldsymbol{X}}_{a_t,u_t})(i)$ is also a martingale. Specially, define $\tilde{\mathcal{F}}_t$ as the $\sigma-$algebra generated by $(vec(\mathring{\boldsymbol{X}}_{a_1,u_1}),\mathcal{A}_1,\eta_1, \ldots, vec(\mathring{\boldsymbol{X}}_{a_t,u_t}),\mathcal{A}_t,\eta_t,\mathcal{A}_{t+1})$. For each dimension $i$ of $vec(\mathring{\boldsymbol{X}}_{a_t,u_t})$, we have $\mathbb{E}[\eta_s vec(\mathring{\boldsymbol{X}}_{a_s,u_s})(i)|\tilde{\mathcal{F}}_{s-1}] = vec(\mathring{\boldsymbol{X}}_{a_s,u_s})(i)\mathbb{E}[\eta_s|\tilde{\mathcal{F}}_{s-1}] = 0$. By using the assumption that $|vec(\mathring{\boldsymbol{X}}_{a_t,u_t})(i)| \leq s_A$, we can obtain, for each $\alpha \in \mathbb{R}$, 
\begin{align*}
    \mathbb{E}\Big[\exp(\alpha \eta_s vec(\mathring{\boldsymbol{X}}_{a_s,u_s})(i))|\tilde{\mathcal{F}}_{s-1} \Big] &\leq \mathbb{E}[\exp(\alpha \eta_s s_A)|\tilde{\mathcal{F}}_{s-1} ] \notag \\
    &\leq \exp(\frac{\alpha^2 s_A^2 \sigma^2}{2}),
\end{align*}
which indicates that $\eta_s vec(\mathring{\boldsymbol{X}}_{a_s,u_s})(i)$ is also a sub-Gaussian random variable with variance proxy $(s_A\sigma)^2$.

Using Theorem~\ref{the:219}, we can obtain
\begin{gather*}
    \mathbb{P}(|\sum_{s=1}^t \eta_s vec(\mathring{\boldsymbol{X}}_{a_s,u_s})(i))| \geq \frac{t\lambda}{2}) \leq 2\exp(-\frac{t \lambda^2}{8s_A^2\sigma^2}).
\end{gather*}
Denote $\lambda = \frac{1}{2}\lambda_t$, the inequality becomes
\begingroup
\footnotesize
\begin{align}
     \mathbb{P}\Big(\frac{2}{t}\norm{\sum_{s=1}^t \eta_s vec(\mathring{\boldsymbol{X}}_{a_s,u_s})}_{\infty} &\leq \frac{\lambda_t}{2}\Big) \geq 1- 2nd\exp(-\frac{t \lambda_t^2}{32s_A^2\sigma^2}) \notag \\
     = 1 - 2\exp&(-\frac{t \lambda_t^2}{32s_A^2\sigma^2}+\log nd).
\end{align}
\endgroup
That completes the proof. 
\end{proof}
%------------------------------
%------------------------------>Lemma
\begin{lemma}
\label{lem:delta-bound}
Let $\widehat{\overline{\boldsymbol{\Sigma}}}_t$ be the empirical covariance matrix of the selected context. Suppose $\widehat{\overline{\boldsymbol{\Sigma}}}_t$ satisfies the compatibility condition with the support $J(vec(\boldsymbol{V}))$ with the compatibility constant $\phi_t$. Then we have
\begin{gather*}
    \mathbb{P} \Big( \norm{\boldsymbol{\Delta}}_1 \leq \frac{4\lambda_t}{\phi_t^2} \Big) \geq 1 - 2\exp \Big(-\frac{t\lambda_t^2}{32\sigma^2s_A^2} + \log (nd)\Big).
\end{gather*}
\end{lemma}
%---------------------------------
\begin{proof}
From the definition of loss function in \eqref{eq:loss} and the update step in Algorithm~\ref{alg:cola}, we have
\begin{gather}
    l_t(\widehat{\boldsymbol{V}}_0^t) + \lambda_t\norm{\widehat{\boldsymbol{V}}_0^t}_1 \leq l_t(vec(\boldsymbol{V})) + \lambda_t\norm{vec(\boldsymbol{V})}_1.
\end{gather}
By taking the expectation over $r_{a_t,u_t} - vec(\mathring{\boldsymbol{X}}_{a_t,u_t})^Tvec(\widehat{\boldsymbol{\Theta}}_t\boldsymbol{W})$, we obtain
\begin{align*}
    l_t(\widehat{\boldsymbol{V}}_0^t) &- \mathbb{E}[l_t(\widehat{\boldsymbol{V}}_0^t)] + \mathbb{E}[l_t(\widehat{\boldsymbol{V}}_0^t)] - \mathbb{E}[l_t(vec(\boldsymbol{V}))] + \lambda_t\norm{\widehat{\boldsymbol{V}}_0^t}_1 \notag \\
    &\leq l_t(vec(\boldsymbol{V})) - \mathbb{E}[l_t(vec(\boldsymbol{V}))] + \lambda_t\norm{vec(\boldsymbol{V})}_1.
\end{align*}
Denoting $\varepsilon_t(vec(\boldsymbol{V})) = l_t(vec(\boldsymbol{V})) -  \mathbb{E}[l_t(vec(\boldsymbol{V}))]$ and $\mathcal{E}(vec(\boldsymbol{V}')) = \mathbb{E}[l_t(vec(\boldsymbol{V}'))] - \mathbb{E}[l_t(vec(\boldsymbol{V}))]$, the previous inequality becomes
\begin{gather}
    \varepsilon_t(\widehat{\boldsymbol{V}}_0^t) + \mathcal{E}(\widehat{\boldsymbol{V}}_0^t)+ \lambda_t\norm{\widehat{\boldsymbol{V}}_0^t}_1 \leq \varepsilon_t(vec(\boldsymbol{V})) + \lambda_t\norm{vec(\boldsymbol{V})}_1. \label{eq:mathcal-e}
\end{gather}
Define the event $\xi_t$ as
\begin{gather*}
    \xi_t = \Big \{| \varepsilon_t(\widehat{\boldsymbol{V}}_0^t) - \varepsilon_t(vec(\boldsymbol{V}))| \leq \frac{1}{2}\lambda_t \norm{\widehat{\boldsymbol{V}}_0^t - vec(\boldsymbol{V})}_1 \Big \}. 
\end{gather*}
We condition on this event in the rest of the proof.

Given the event $\xi_t$ and inequality~\eqref{eq:mathcal-e}, we have 
\begin{gather}
    2\mathcal{E}(\widehat{\boldsymbol{V}}_0^t) \leq 2\lambda_t(\norm{vec(\boldsymbol{V})}_1 - \norm{\widehat{\boldsymbol{V}}_0^t}_1) + \lambda_t \norm{\widehat{\boldsymbol{V}}_0^t - vec(\boldsymbol{V})}_1. \label{eq:mathcal-e-event}
\end{gather}
According to the triangle inequality,
\begin{align}
    &\norm{\widehat{\boldsymbol{V}}_0^t}_1 = \norm{\widehat{\boldsymbol{V}}_0^t(\mathcal{J})}_1 + \norm{\widehat{\boldsymbol{V}}_0^t(\mathcal{J}^c)}_1 \notag \\
    &\geq \norm{vec(\boldsymbol{V})(\mathcal{J})}_1 - \norm{\widehat{\boldsymbol{V}}_0^t(\mathcal{J})-vec(\boldsymbol{V})(\mathcal{\mathcal{J}})}_1 + \norm{\widehat{\boldsymbol{V}}_0^t(\mathcal{J}^c)}_1, \label{eq:triangle}
\end{align}
where $\mathcal{J}^c$ is the complement of $\mathcal{J}$. In addition, we have
\begin{align}
    \norm{\widehat{\boldsymbol{V}}_0^t - vec(\boldsymbol{V})}_1 
    &= \norm{(\widehat{\boldsymbol{V}}_0^t - vec(\boldsymbol{V}))(\mathcal{J})}_1 + \norm{\widehat{\boldsymbol{V}}_0^t(\mathcal{J}^c)}_1. \label{eq:hat-v-diff}
\end{align}
Combining inequalities~\eqref{eq:mathcal-e-event}, \eqref{eq:triangle} and \eqref{eq:hat-v-diff}, we can obtain
\begin{align}
    2\mathcal{E}(\widehat{\boldsymbol{V}}_0^t) 
    \leq - \lambda_t \norm{\widehat{\boldsymbol{V}}_0^t(\mathcal{J}^c)}_1 + 3\lambda_t \norm{\widehat{\boldsymbol{V}}_0^t(\mathcal{J})-vec(\boldsymbol{V})(\mathcal{J})}_1 . \label{eq:e-vecv}
\end{align}
According to the compatibility condition in Definition~\ref{def:cc}, we can obtain
\begin{align*}
   &\norm{\widehat{\boldsymbol{V}}_0^t(\mathcal{J})-vec(\boldsymbol{V})(\mathcal{J})}_1^2 \notag \\
   &\quad \quad \leq \frac{(\widehat{\boldsymbol{V}}_0^t-vec(\boldsymbol{V}))^T\widehat{\overline{\boldsymbol{\Sigma}}}_t(\widehat{\boldsymbol{V}}_0^t-vec(\boldsymbol{V}))}{\phi_t^2}.
\end{align*}
Thus, with \eqref{eq:e-vecv}, we arrive at 
\begin{align}
    &2\mathcal{E}(\widehat{\boldsymbol{V}}_0^t) + \lambda_t \norm{\widehat{\boldsymbol{V}}_0^t - vec(\boldsymbol{V})}_1 \notag \\
    &= 2\mathcal{E}(\widehat{\boldsymbol{V}}_0^t) + \lambda_t \norm{(\widehat{\boldsymbol{V}}_0^t - vec(\boldsymbol{V}))(\mathcal{J})}_1 + \lambda_t \norm{\widehat{\boldsymbol{V}}_0^t(\mathcal{J}^c)}_1 \notag \\
    &\leq 4 \lambda_t \norm{\widehat{\boldsymbol{V}}_0^t(\mathcal{J})-vec(\boldsymbol{V})(\mathcal{J})}_1 \notag \\
    &\leq \frac{4\lambda_t}{\phi_t}\sqrt{(\widehat{\boldsymbol{V}}_0^t-vec(\boldsymbol{V}))^{\top}\widehat{\overline{\boldsymbol{\Sigma}}}_t(\widehat{\boldsymbol{V}}_0^t-vec(\boldsymbol{V}))} \notag \\
    &\leq (\widehat{\boldsymbol{V}}_0^t-vec(\boldsymbol{V}))^{\top}\widehat{\overline{\boldsymbol{\Sigma}}}_t(\widehat{\boldsymbol{V}}_0^t-vec(\boldsymbol{V})) + \frac{4\lambda_t^2}{\phi_t^2} \notag \\
    &\leq 2 \mathcal{E}(\widehat{\boldsymbol{V}}_0^t) + \frac{4\lambda_t^2}{\phi_t^2},
\end{align}
where the first inequality is based on \eqref{eq:hat-v-diff}, the second inequality is based on the compatibility condition, the third inequality uses $4uv \leq u^2+4v^2$, and the last inequality is due to the following Lemma~\ref{lem:ehat-v}.

Thus, using Lemma~\ref{lem:epsilon-v-p} yields
\begin{gather}
    \norm{\widehat{\boldsymbol{V}}_0^t - vec(\boldsymbol{V})}_1 \leq \frac{4\lambda_t}{\phi_t^2},
\end{gather}
with probability at least $1 - 2\exp \Big(-\frac{t\lambda_t^2}{32\sigma^2s_A^2} + \log nd \Big)$. That completes the proof of Lemma~\ref{lem:delta-bound}.
\end{proof}
%-------------------------
%-------------------------->Lemma
\begin{lemma}
\label{lem:sigma-norm}
Let $C_0 = \min \{\frac{1}{2},\frac{\phi_0^2}{512s_0s_A^2\nu\omega_{\mathcal{X}}}\}$. We have, for all $t \geq \frac{2\log(2d^2)}{C_0^2}$,
\begin{gather}
\mathbb{P}(\frac{1}{2s_A^2}\norm{\widehat{\overline{\boldsymbol{\Sigma}}}_t- \overline{\boldsymbol{\Sigma}}_t}_{\infty} \geq \frac{\phi^2(\overline{\boldsymbol{\Sigma}}_t,\mathcal{J})}{64s_0s_A^2}) \leq \exp(-\frac{tC_0^2}{2})   
\end{gather}
\end{lemma}
%------------------------------
%--------------------------------->Proof
\begin{proof}
Let us define $\gamma_t^{ij}(vec(\mathring{\boldsymbol{X}}_{a_s,u_s}))$ as
\begin{align*}
    \gamma_t^{ij}(vec(\mathring{\boldsymbol{X}}_{a_s,u_s})) &= \frac{1}{2s_A^2}\Big (vec(\mathring{\boldsymbol{X}}_{a_s,u_s})(i)vec(\mathring{\boldsymbol{X}}_{a_s,u_s})(j) \notag \\
    &- \mathbb{E}[vec(\mathring{\boldsymbol{X}}_{a_s,u_s})(i)vec(\mathring{\boldsymbol{X}}_{a_s,u_s})(j) | \mathcal{F}'_{t-1}] \Big),
\end{align*}
where $\mathcal{F}'_t$ is the $\sigma$-algebra generated by random variables $(\mathcal{A}_{u_1,1},vec(\mathring{\boldsymbol{X}}_{a_1,u_1}),r_{a_1,u_1},\ldots,\mathcal{A}_{t-1},vec(\mathring{\boldsymbol{X}}_{a_{t-1},u_{t-1}}),\\
r_{a_{t-1},u_{t-1}},\mathcal{A}_{u_t,t})$. We have 
\begin{align*}
\frac{1}{2s_A^2}\norm{\widehat{\overline{\boldsymbol{\Sigma}}}_t - \overline{\boldsymbol{\Sigma}}_t}_{\infty} = \max_{1\leq i \leq j \leq nd}|\frac{1}{t}\sum_{s=1}^t \gamma_s^{ij}(vec(\mathring{\boldsymbol{X}}_{a_s,u_s}))|,
\end{align*}
$\mathbb{E}[\gamma_t^{ij}(vec(\mathring{\boldsymbol{X}}_{a_s,u_s}))|\mathcal{F}'_{t-1}] = 0$ and $\mathbb{E}[|\gamma_t^{ij}(vec(\mathring{\boldsymbol{X}}_{a_s,u_s}))|^m | \mathcal{F}'_{t-1}] \leq 1$ for all integer $m \geq 2$. Therefore, we can apply Lemma~\ref{lem:bl-ineq}:
\begin{align*}
\mathbb{P}\Big( \frac{1}{2s_A^2} \norm{\widehat{\overline{\boldsymbol{\Sigma}}}_t - \overline{\boldsymbol{\Sigma}}_t}_{\infty} \geq x& + \sqrt{2x} + \sqrt{\frac{4\log(2n^2d^2)}{t}} \notag \\
&+ \frac{2\log(2n^2d^2)}{t} \Big) \leq \exp(-\frac{tx}{2}).     
\end{align*}
For all $t \geq \frac{2\log(2n^2d^2)}{C_0^2}$ with $C_0 = \min \{\frac{1}{2}, \frac{\phi_0^2}{256s_0s_A^2\nu \omega_{\mathcal{X}}}\}$, taking $x = C_0^2$,
\begin{align}
x +& \sqrt{2x} + \sqrt{\frac{4\log(2n^2d^2)}{t}} \leq 2C_0^2 + 2\sqrt{2}C_0 \notag  \\
&\quad \leq 4C_0 \leq \frac{\phi_0^2}{128s_0s_A^2\nu \omega_{\mathcal{X}}} \leq \frac{\phi^2(\overline{\boldsymbol{\Sigma}}_t,J)}{64s_0s_A^2}
\end{align}
In summary, for all $ t \geq \frac{2\log(2n^2d^2)}{C_0^2}$, we have
\begin{align}
\mathbb{P}\Big( \frac{1}{2s_A^2} \norm{\widehat{\overline{\boldsymbol{\Sigma}}}_t - \overline{\boldsymbol{\Sigma}}_t}_{\infty}  \geq \frac{\phi^2(\overline{\boldsymbol{\Sigma}}_t,\mathcal{J})}{64s_0s_A^2}\ \Big) 
&\leq \exp(-\frac{tC_0^2}{2})
\end{align}
That completes the proof.
\end{proof}
%---------------------------------
%--------------------------------->Lemma
\begin{lemma}
\label{lem:phi-sigma}
Let $C_0 = \min \{\frac{1}{2},\frac{\phi_0^2}{512s_0s_A^2\nu \omega_{\mathcal{X}}}\}$. For all $t \geq \frac{2\log(2n^2d^2)}{C_0^2}$, we have 
\begin{gather}
    \mathbb{P}(\phi^2(\widehat{\overline{\boldsymbol{\Sigma}}}_t, \mathcal{J}) \geq \frac{\phi_0^2}{4\nu \omega_{\mathcal{X}}}) \geq 1 - \exp\Big(-\frac{tC_0^2}{2} \Big)
\end{gather}
\end{lemma}
%-------------------------->Proof
\begin{proof}
In Definition~\ref{def:cc}, the considered covariance matrix $\overline{\boldsymbol{\Sigma}} \in \mathbb{R}^{nd \times nd}$ is block-diagonal consisting of $n$ blocks, one for each different user. Define the filtration $(\overline{\mathcal{F}}_t)_{t \geq 0}$ as $\overline{\mathcal{F}}_t = \sigma(\mathring{\boldsymbol{X}}_{a_1,u_1},\eta_1, \ldots, \mathring{\boldsymbol{X}}_{a_t,u_t})$. First, consider the single-user adapted matrix
\begin{gather}
 \boldsymbol{\Sigma}_{i,t} = \frac{1}{|\mathcal{T}_{i,t}|}\sum_{s \in \mathcal{T}_{i,t}} \mathbb{E}[\boldsymbol{x}_{i,s}\boldsymbol{x}_{i,s}^{\top}|\overline{\mathcal{F}}_{s-1}], \quad i \in \mathcal{U}
\end{gather}
where $\mathcal{T}_{i,t}$ collects the round that user $i$ is selected. The multi-user matrix is then $\overline{\boldsymbol{\Sigma}}_t = diag(\boldsymbol{\Sigma}_{1,t}, \ldots, \boldsymbol{\Sigma}_{n,t}) \in \mathbb{R}^{nd \times nd}$. According to Lemma 10 of \cite{oh2021sparsity} and Assumption~\ref{assump:armsd}, we have
\begin{align}
    \overline{\boldsymbol{\Sigma}}_t &= \frac{1}{t}\sum_{s=1}^t\mathbb{E}[vec(\mathring{\boldsymbol{X}}_{a_s,u_s})vec(\mathring{\boldsymbol{X}}_{a_s,u_s})^{\top}|\overline{\mathcal{F}}_{s-1}] \notag \\
    & \succeq (2\nu \omega_{\mathcal{X}})^{-1}\overline{\boldsymbol{\Sigma}},
    \label{eq:sigma-cc}
\end{align}
where $vec(\mathring{\boldsymbol{X}}_{a_s,u_s}) \in \mathbb{R}^{nd}$. According to Definition~\ref{def:cc},
\begin{gather}    \frac{\norm{\boldsymbol{\Delta}}^2_{\overline{\boldsymbol{\Sigma}}_t}}{\norm{\boldsymbol{\Delta}(\mathcal{J})}^2} \geq \frac{\norm{\boldsymbol{\Delta}}^2_{\overline{\boldsymbol{\Sigma}}}}{2\nu \omega_{\mathcal{X}}\norm{\boldsymbol{\Delta}(\mathcal{J})}^2}  \geq \frac{\phi^2(\overline{\boldsymbol{\Sigma}}, \mathcal{J})}{2\nu \omega_{\mathcal{X}}} \geq \frac{\phi_0^2}{2\nu \omega_{\mathcal{X}}}
\end{gather}
Hence, $\overline{\boldsymbol{\Sigma}}_t$ satisfies the compatibility condition with $\phi^2(\overline{\boldsymbol{\Sigma}}_t,\mathcal{J}) = \frac{\phi_0^2}{2\nu \omega_{\mathcal{X}}}$. Furthermore, we have the following adaptive matrix concentration results for $\widehat{\overline{\boldsymbol{\Sigma}}}_t$:

By using Lemma~\ref{lem:sigma-norm} and Corollary~\ref{cor:68}, for all $ t \geq \frac{2\log(2n^2d^2)}{C_0^2}$, 
\begin{align}
\phi^2(\widehat{\overline{\boldsymbol{\Sigma}}}_t,\mathcal{J}) \geq \frac{\phi^2(\overline{\boldsymbol{\Sigma}}_t,\mathcal{J})}{2} \geq \frac{\phi_0^2}{4\nu\omega_{\mathcal{X}}},
\end{align}
with probability at least $1 - \exp(-\frac{tC_0^2}{2})$. That completes the proof. 
\end{proof}
%-------------------------
%----------------------------->Lemma
\begin{lemma}[Minimal Eigenvalue of the Empirical Gram Matrix]\label{lem:min-eigen}
Under Assumption~\ref{assump:spd}, we have $\mathbb{P}(\lambda_{\min}(\widehat{\overline{\boldsymbol{\Sigma}}}_{t,\hat{J}}) \geq \frac{\beta}{4\nu\omega_{\mathcal{X}}}|\mathcal{D}_t) \geq 1 - \exp(\log(s_0+ \frac{4\nu\omega_{\mathcal{X}}\sqrt{s_0}}{\phi_0^2}) - \frac{t\beta}{20s_A^2\nu \omega_{\mathcal{X}}(s_0+(4\nu\omega_{\mathcal{X}}\sqrt{s_0})/\phi_0^2)})$.  
\end{lemma}
%------------------------------
\begin{proof}
For a fixed set $\hat{\mathcal{J}}$, we first define the adapted Gram matrix on the estimated support as
\begin{gather}
\overline{\boldsymbol{\Sigma}}_{t,\hat{\mathcal{J}}} = \frac{1}{t}\sum_{s=1}^t\mathbb{E}[(vec(\mathring{\boldsymbol{X}}_{a_s,u_s})(\hat{\mathcal{J}}))(vec(\mathring{\boldsymbol{X}}_{a_s,u_s})(\hat{\mathcal{J}}))^{\top}|\mathcal{F}'_{s-1}]  
\end{gather}
The following Lemma characterizes the expected Gram matrix generated by the algorithm.
%------------------------>Lemma
\begin{lemma}\label{lem:sigma-support}
Fix $\hat{\mathcal{J}}$ such that $\mathcal{J} \subset \hat{\mathcal{J}}$ and $|\hat{\mathcal{J}}| \leq s_0+(4\nu\omega_{\mathcal{X}}\sqrt{s_0})/\phi_0^2$. Under Assumption~\ref{assump:armsd} and Definition~\ref{def:cc}, we have
\begin{gather}
    \overline{\boldsymbol{\Sigma}}_{t,\hat{J}} \succeq (2\nu \omega_{\mathcal{X}})^{-1}\overline{\boldsymbol{\Sigma}}_{\hat{J}}.
\end{gather}
\end{lemma}
%----------------------------
The proof of Lemma~\ref{lem:sigma-support} is almost the same as the proof of \eqref{eq:sigma-cc} in Lemma~\ref{lem:phi-sigma}, thus we omit the proof here.

By Assumption~\ref{assump:spd} and the construction of our algorithm, under the event $\mathcal{D}_t$
\begin{align}
\lambda_{\min}(\overline{\boldsymbol{\Sigma}}_{t,\hat{\mathcal{J}}}) 
&\geq \frac{\beta}{2\nu\omega_{\mathcal{X}}},
\end{align}
where the first inequality uses the concavity of $\lambda_{\min}(\cdot)$ over the positive semi-definite matrices. Next, we prove the upper bound on the largest eigenvalue of $(vec(\mathring{\boldsymbol{X}}_{a_s,u_s})(\hat{\mathcal{J}}))(vec(\mathring{\boldsymbol{X}}_{a_s,u_s})(\hat{\mathcal{J}}))^{\top}$
\begin{align}
   & \lambda_{\max}((vec(\mathring{\boldsymbol{X}}_{a_s,u_s})(\hat{\mathcal{J}}))(vec(\mathring{\boldsymbol{X}}_{a_s,u_s})(\hat{\mathcal{J}}))^{\top}) \notag \\
   &= \max_{\norm{\boldsymbol{l}}_2 = 1} \boldsymbol{l}^{\top} (vec(\mathring{\boldsymbol{X}}_{a_s,u_s})(\hat{\mathcal{J}}))(vec(\mathring{\boldsymbol{X}}_{a_s,u_s})(\hat{\mathcal{J}}))^{\top} \boldsymbol{l} \notag \\
    &\leq \max_{\norm{\boldsymbol{l}}_2 = 1} \norm{\boldsymbol{l}}_1^2 \norm{(vec(\mathring{\boldsymbol{X}}_{a_s,u_s})(\hat{\mathcal{J}}))}_{\infty}^2 \notag \\
    &\leq |\hat{\mathcal{J}}|s_A^2 \leq (s_0 + (4\nu\omega_{\mathcal{X}}\sqrt{s_0})/\phi_0^2)s_A^2,
\end{align}
where the first inequality uses H\"older's inequality. 

Using Theorem~\ref{the:mc} with $R= (s_0 + (4\nu\omega_{\mathcal{X}}\sqrt{s_0})/\phi_0^2)s_A^2$, $X_s = (vec(\mathring{\boldsymbol{X}}_{a_s,u_s})(\hat{\mathcal{J}}))(vec(\mathring{\boldsymbol{X}}_{a_s,u_s})(\hat{\mathcal{J}}))^{\top}$, $Y = t\widehat{\overline{\boldsymbol{\Sigma}}}_{t,\hat{J}}$, $W = t\overline{\boldsymbol{\Sigma}}_{t,\hat{J}}$, $\delta = \frac{1}{2}$ and $\mu = t \frac{\beta}{2\nu \omega_{\mathcal{X}}}$.
\begin{align*}
    &\mathbb{P}(\lambda_{\min}(t\widehat{\overline{\boldsymbol{\Sigma}}}_{t,\hat{J}}) \leq \frac{t}{2}\frac{\beta}{2\nu\omega_{\mathcal{X}}} \text{ and } \lambda_{\min}(t\overline{\boldsymbol{\Sigma}}_{t,\hat{J}}) \geq \frac{\beta t}{2\nu\omega_{\mathcal{X}}}) \leq \notag \\
    & \exp(\log(s_0 + \frac{4\nu\omega_{\mathcal{X}}\sqrt{s_0}}{\phi_0^2}) - \frac{t\beta}{20s_A^2\nu\omega_{\mathcal{X}}(s_0 + \frac{4\nu\omega_{\mathcal{X}}\sqrt{s_0}}{\phi_0^2})}),
\end{align*}
where the last inequality uses $-0.5-0.5\log(0.5) \leq -\frac{1}{10}$. That completes the proof. 
\end{proof}
%---------------------------------------------
%-------------------------------->Lemma
\begin{lemma}[Estimation after Thresholding]\label{lem:eat}
Let $s' = s_0+ \frac{4\nu\omega_{\mathcal{X}}\sqrt{s_0}}{\phi_0^2}$. Under Assumption~\ref{assump:arfv} and Remark~\ref{rem:s}, we have, for all $x,\lambda >0$: $\mathbb{P}\Big(\frac{1}{\lambda t} \norm{\mathcal{X}(\hat{\mathcal{J}})^{\top}\eta}_2 \geq x \text{ and } \lambda_{\min}(\widehat{\overline{\boldsymbol{\Sigma}}}_{t,\hat{J}}) \geq \lambda | \mathcal{D}_t \Big) \leq 2s'\exp(-\frac{\lambda^2tx^2}{2\sigma^2s_A^2s'})$. 
\end{lemma}
%----------------------------------
%---------------------------------->proof
\begin{proof}
Denote $\hat{\mathcal{J}} = \hat{\mathcal{J}}_1^t$ and $\eta = (\eta_1, \eta_2, \ldots, \eta_t)$, $\mathcal{X}(\hat{\mathcal{J}})$ denotes the $t\times |\hat{\mathcal{J}}|$ submatrix of $\mathcal{X} \in \mathbb{R}^{t \times nd}$. Conditioning on a fixed $\hat{\mathcal{J}}$, based on Theorem~\ref{the:219}, we can derive 
\begin{align}
&\mathbb{P}\Big(\frac{1}{\lambda t} \norm{\mathcal{X}(\hat{\mathcal{J}})^{\top}\eta}_2 \geq x \text{ and } \lambda_{\min}(\widehat{\overline{\boldsymbol{\Sigma}}}_{t,\hat{\mathcal{J}}}) \geq \lambda \Big) \notag \\
&\leq \mathbb{P}\Big(\norm{\mathcal{X}(\hat{\mathcal{J}})^{\top}\eta}_2  \geq \lambda tx | \lambda_{\min}(\widehat{\overline{\boldsymbol{\Sigma}}}_{t,\hat{J}}) \geq \lambda \Big)\mathbb{P}\Big( \lambda_{\min}(\widehat{\overline{\boldsymbol{\Sigma}}}_{t,\hat{J}}) \geq \lambda \Big) \notag \\
&\leq \mathbb{P}\Big(\norm{\mathcal{X}(\hat{\mathcal{J}})^{\top}\eta}_2  \geq \lambda tx \Big) \notag \\
&\leq \sum_{i \in \hat{J}}\mathbb{P}\Big ( |\sum_{s=1}^t \eta_s (vec(\mathring{\boldsymbol{X}}_{a_s,u_s})(i))| \geq \frac{\lambda tx}{\sqrt{s_0+\frac{4\nu\omega_{\mathcal{X}}\sqrt{s_0}}{\phi_0^2}}}\Big) \notag \\
&\leq 2 s_0+\frac{4\nu\omega_{\mathcal{X}}\sqrt{s_0}}{\phi_0^2}\exp(-\frac{\lambda^2t x^2}{2\sigma^2s_A^2 s_0+\frac{4\nu\omega_{\mathcal{X}}\sqrt{s_0}}{\phi_0^2}})
\end{align}
That completes the proof. 
\end{proof}
%-----------------------------------
%-----------------------
\begin{lemma}\label{lem:als-est-err}
   When the ALS step is properly initialized, $\epsilon_v = vec(\mathring{\boldsymbol{X}}_{a_s,u_s})^{\top}(vec(\widehat{\boldsymbol{V}}_s) - vec(\boldsymbol{V}))$ is zero mean $2s_xs_v$-sub-Gaussian, conditioned on $\mathcal{F}_t$. Similarly, $\epsilon_{\theta} = vec(\mathring{\boldsymbol{X}}_{a_s,u_s}\boldsymbol{W}^{\top})^{\top}(vec(\widehat{\boldsymbol{\Theta}}_s) - vec(\boldsymbol{\Theta}))$ is zero-mean $2s_xs_{\theta}$-sub-Gaussian, conditioned on $\mathcal{F}_t$. 
\end{lemma}
%------------------------------
\begin{proof}
Lemma~\ref{lem:als-est-err} is adapted from Lemma F.1 in \cite{li2022asynchronous}. Here we summarize the proof with variables in NELA briefly.

When conditioning on $\mathcal{F}_t$, $\epsilon_v$ is a random variable bounded by $[-2s_xs_v,2s_xs_v]$. Therefore, $\epsilon_v$ is sub-Gaussian. For $\epsilon_{\theta}$, the proof follows the same procedure. 
\end{proof}
%------------------

%------------------------
\subsection{Proof of Lemma~\ref{lem:ucb}}\label{app:lem-ucb}
\begin{proof}
Define 
\begin{align*}
    \boldsymbol{S}_t &= \sum_{s=1}^t vec(\mathring{\boldsymbol{X}}_{a_s,u_s}\boldsymbol{W}^{\top})\eta_s,
\end{align*}
then we can obtain
\begin{align*}
    &\norm{vec(\widehat{\boldsymbol{\Theta}}_{t}) - vec(\boldsymbol{\Theta})}_{\boldsymbol{A}_t} \\
    &\quad \quad \leq \norm{\boldsymbol{S}_t}_{\boldsymbol{A}_t^{-1}} + \sqrt{\lambda_1}\norm{vec(\boldsymbol{\Theta})}  +\norm{\boldsymbol{Q}_{v,t}}_{\boldsymbol{A}_t^{-1}},
\end{align*}
where $\boldsymbol{Q}_{v,t} = \sum_{s=1}^t vec(\mathring{\boldsymbol{X}}_{a_s,u_s}\boldsymbol{W}^{\top}) vec(\mathring{\boldsymbol{X}}_{a_s,u_s})^{\top}(vec(\widehat{\boldsymbol{V}}_s) - vec(\boldsymbol{V}))$. 

We first bound the first term. Because $\boldsymbol{S}_t$ is a martingale, according to Theorem 1 and Theorem 2 in \cite{abbasi2011improved}, we have
\begin{gather*}
    \norm{\boldsymbol{S}_t}_{\boldsymbol{A}_t^{-1}} \leq \sigma \sqrt{2 \log (\frac{\det(\boldsymbol{A}_t)^{1/2}\det(\boldsymbol{A}_0)^{-1/2}}{\delta})},
\end{gather*}
with probability at least $1 - \delta$, for any $\delta > 0$. Besides, we have $\det(\boldsymbol{A}_0) = \det(\lambda_1 \boldsymbol{I}) \leq \lambda_1^{dn}$. 

Now it remains to bound the term $\norm{\boldsymbol{Q}_{v,t}}_{\boldsymbol{A}_t^{-1}}$, according to Lemma~\ref{lem:als-est-err}, we can bound
\begin{align*}
    \norm{\boldsymbol{Q}_{v,t}}_{\boldsymbol{A}_t^{-1}} \leq 2s_xs_A\sqrt{2 \log (\frac{\det(\boldsymbol{A}_t)^{1/2}\det(\boldsymbol{A}_0)^{-1/2}}{\delta})}.
\end{align*}
Putting all these together, we have
\begin{align}
    \norm{vec(\widehat{\boldsymbol{\Theta}}_{t}) - vec(\boldsymbol{\Theta})}_{\boldsymbol{A}_t} &\leq (\sigma +2s_xs_v)\sqrt{\frac{\det(\boldsymbol{A}_t)}{\lambda_1^{dn} \delta^2}} + \sqrt{\lambda_1}s_{\theta} . \label{eq:vecb}
\end{align}
In addition, we have $\Tr(\boldsymbol{A}_t) \leq \lambda_1 dn + \sum_{\tau=1}^t\sum_{j=1}^n W(u_{\tau},j)^2s_A^2 d^2$, we have $\det(\boldsymbol{A}_t) \leq (\frac{\Tr(\boldsymbol{A}_t)}{dn})^{dn} \leq (\lambda_1 + \frac{\sum_{\tau=1}^t\sum_{j=1}^n W(u_{\tau},j)^2s_A^2d^2}{dn})^{dn}$. Therefore, \eqref{eq:vecb} can further bounded by
\begin{align*}
    &\norm{vec(\widehat{\boldsymbol{\Theta}}_{t}) - vec(\boldsymbol{\Theta})}_{\boldsymbol{A}_t} \leq \sqrt{\lambda_1}s_{\theta} + (\sigma + 2s_xs_v)\times \notag \\
    & \sqrt{dn\log(1+\frac{\sum_{\tau=1}^t\sum_{j=1}^n W(u_{\tau},j)^2s_A^2d^2}{dn \lambda_1}) - 2\log \delta}.
\end{align*}
\end{proof}
%------------------------------------
%---------------
\subsection{Proof of Lemma~\ref{lem:support-recover}}\label{app:lem-support-recover}
%------------------------->Proof
\begin{proof}
Define $\boldsymbol{\Delta}_t = \widehat{\boldsymbol{V}}_0^t - vec(\boldsymbol{V})$ and define the event $\mathcal{H}_t$ as
\begin{gather}
    \mathcal{H}_t = \Big \{ \norm{\boldsymbol{\Delta}_t}_1 \leq \frac{4s_0\lambda_t}{\phi_t^2} \Big\},
\end{gather}
where $\phi_t$ is the compatibility constant of $\widehat{\overline{\boldsymbol{\Sigma}}}_t$ defined in Lemma~\ref{lem:delta-bound}. 

For the rest of the proof, we assume that the event $\mathcal{H}_t$ holds. Besides, we have
\begin{align*}
    \norm{\boldsymbol{\Delta}_t}_1 &\geq \norm{\boldsymbol{\Delta}_t(\mathcal{J}^c)}_1  \geq \sum_{j \in \mathcal{J}^c \cap \hat{\mathcal{J}}_0^t} |\widehat{\boldsymbol{V}}_0^t(j)| \geq |\hat{\mathcal{J}}_0^t \backslash \mathcal{J}|4\lambda_t , 
\end{align*}
where in the last inequality we use the construction of $\hat{\mathcal{J}}_0^t$ in the algorithm. We have 
\begin{gather}
    |\hat{\mathcal{J}}_0^t \backslash \mathcal{J}| \leq \frac{\norm{\boldsymbol{\Delta}_t}_1}{4\lambda_t} \leq \frac{s_0}{\phi_t^2},
\end{gather}
where in the last inequality we used the definition of $\mathcal{H}_t$. We have $\forall j \in J$,
\begin{align}
    |\widehat{\boldsymbol{V}}_0^t(j)|  \geq  V_{\min} - \norm{\boldsymbol{\Delta}(\mathcal{J})}_1 \geq V_{\min} - \frac{4s_0\lambda_t}{\phi_t^2}.
\end{align}
Therefore, when $t$ is large enough so that $4\lambda_t \leq V_{\min} - \frac{4s_0\lambda_t}{\phi_t^2}$, we have $\mathcal{J} \subset \hat{\mathcal{J}}_0^t$. Similarly, when $t$ is large enough so that $4\lambda_t \sqrt{(1+\frac{1}{\phi_t^2})s_0} \leq V_{\min} - \frac{4s_0\lambda_t}{\phi_t^2}$, it holds that $\mathcal{J} \subset \hat{\mathcal{J}}_0^t$. From the construction of $\hat{\mathcal{J}}_1^t$ in the algorithm, it also holds that $\hat{\mathcal{J}}_1^t \subset \hat{\mathcal{J}}_0^t$. Therefore, 
\begin{align}
     \norm{\boldsymbol{\Delta}_t}_1 &\geq \sum_{j \in \hat{\mathcal{J}}_0^t \backslash \mathcal{J}}|\widehat{\boldsymbol{V}}_0^t(j)| \geq |\hat{\mathcal{J}}_1^t \backslash \mathcal{J}| 4\lambda_t \sqrt{|\hat{\mathcal{J}}_0^t|},
\end{align}
and 
\begin{align}
    |\hat{\mathcal{J}}_1^t \backslash \mathcal{J}| &\leq \frac{\norm{\boldsymbol{\Delta}_t}_1}{4\lambda_t\sqrt{|\hat{\mathcal{J}}_0^t|}} \leq \frac{1}{4\lambda_t\sqrt{|\hat{\mathcal{J}}_0^t|}} \frac{4s_0\lambda_t}{\phi_t^2} \leq \frac{\sqrt{s_0}}{\phi_t^2}. 
\end{align}
The condition that $4\lambda_t\sqrt{(1+\frac{1}{\phi_t^2})s_0} \leq vec(\boldsymbol{V})_{\min} - \frac{4s_0\lambda_t}{\phi_t^2}$ is equivalent to $4\lambda_t(\sqrt{(1+\frac{1}{\phi_t^2})s_0} + \frac{s_0}{\phi_t^2}) \leq vec(\boldsymbol{V})_{\min}$. Based on the auxiliary Lemmas~\ref{lem:delta-bound} and \ref{lem:phi-sigma}, we can guarantee the estimation of the support by substituting $\phi_t^2 = \frac{\phi_0^2}{4\nu\omega_{\mathcal{X}}}$. In the worst case, each anomaly in $\mathcal{U}_t$ contains at least one non-zero dimension, which indicates $\mathbb{P}(\tilde{\mathcal{U}} \subset \tilde{\mathcal{U}}_t$ and $|\tilde{\mathcal{U}}_t \backslash \tilde{\mathcal{U}}| \leq \frac{4\nu\omega_{\mathcal{X}}\sqrt{s_0}}{\phi_0^2}) \geq 1 - 2\exp(-\frac{t\lambda_t^2}{32\sigma^2s_A^2}+\log nd)-\exp(-\frac{tC_0^2}{2})$. 
\end{proof}
%-------------------------
%--------------------------->
\subsection{Proof of Lemma~\ref{lem:ins-rgt}}\label{app:lem-ins-rgt}
\begin{proof}
\begin{align}
   \mathbb{E}[R_t] &=  \langle vec(\mathring{\boldsymbol{X}}_{a_t^*,u_{t}}),vec(\boldsymbol{\Theta}\boldsymbol{W}) + vec(\boldsymbol{V}) \rangle \notag \\
   & \quad - \langle vec(\mathring{\boldsymbol{X}}_{a_t,u_t}),vec(\boldsymbol{\Theta}\boldsymbol{W}) + vec(\boldsymbol{V}) \rangle \notag \\
   &\leq \underbrace{ \alpha_t \sqrt{vec(\mathring{\boldsymbol{X}}_{a_t,u_t}\boldsymbol{W}^{\top})\boldsymbol{A}_{t-1}^{-1}vec(\mathring{\boldsymbol{X}}_{a_t,u_t}\boldsymbol{W}^{\top})^{\top}}}_{\text{term (a)}} \notag \\
    + &\underbrace{\langle vec(\mathring{\boldsymbol{X}}_{a_t^*,u_t}) - vec(\mathring{\boldsymbol{X}}_{a_t,u_t}),   vec(\boldsymbol{\Theta}\boldsymbol{W}) -vec(\widehat{\boldsymbol{\Theta}}_t\boldsymbol{W}) \rangle}_{\text{term (b)}} \notag \\
   & + \underbrace{\langle vec(\mathring{\boldsymbol{X}}_{a_t^*,u_t}) - vec(\mathring{\boldsymbol{X}}_{a_t,u_t}),    vec(\boldsymbol{V}) - vec(\widehat{\boldsymbol{V}}_t) \rangle}_{\text{term (c)}}  \label{eq:rgt-ins-term}
\end{align}
where the inequality is obtained according to our action selection strategy.

Next, we have
\begin{align*}
    &\text{term (b)}\notag \\
    &= \langle vec(\mathring{\boldsymbol{X}}_{a_t^*,u_t}) - vec(\mathring{\boldsymbol{X}}_{a_t,u_t}),   vec(\boldsymbol{\Theta}\boldsymbol{W})  -vec(\widehat{\boldsymbol{\Theta}}_t\boldsymbol{W}) \rangle \notag \\
    &\leq \alpha_t \Big(\norm{vec(\mathring{\boldsymbol{X}}_{a_t,u_t}\boldsymbol{W}^{\top})}_{\boldsymbol{A}_{t-1}^{-1}} +  \norm{vec(\mathring{\boldsymbol{X}}_{a_t^*,u_t}\boldsymbol{W}^{\top})}_{\boldsymbol{A}_{t-1}^{-1}}\Big).
\end{align*}
To bound term (c) of \eqref{eq:rgt-ins-term}, we first define the event $\mathcal{L}_t^{\frac{\beta}{4\nu\omega_{\mathcal{X}}}} = \Big\{ \lambda_{\min}(\widehat{\overline{\boldsymbol{\Sigma}}}_{t,\hat{J}}) \geq \frac{\beta}{4\nu\omega_{\mathcal{X}}} \Big\}$, we have 
\begin{align}
    \text{term (c)} &=  \langle vec(\mathring{\boldsymbol{X}}_{a_t^*,u_t}) - vec(\mathring{\boldsymbol{X}}_{a_t,u_t}),    vec(\boldsymbol{V}) - vec(\widehat{\boldsymbol{V}}_t) \rangle \notag \\
    &\leq 2s_A\mathbb{E}\Big[\norm{vec(\boldsymbol{V}) - vec(\widehat{\boldsymbol{V}}_t)}_2| \mathcal{D}_t,\mathcal{L}_t^{\frac{\beta}{4\nu\omega_{\mathcal{X}}}}\Big]\notag \\
    &\quad +2s_1s_A\Big (\mathbb{P}(\mathcal{D}_t^c) + \mathbb{P}((\mathcal{L}_t^{\frac{\beta}{4\nu\omega_{\mathcal{X}}}})^c |\mathcal{D}_t) \Big) 
\end{align}
Specially, use $\hat{\mathcal{J}} = \hat{\mathcal{J}}_1^t$ and $\eta = (\eta_1, \eta_2, \ldots, \eta_t)$. Assume $\lambda_{\min}(\widehat{\overline{\boldsymbol{\Sigma}}}_{t,\hat{J}}) \geq \lambda$. We have
\begin{align}
 &\norm{vec(\widehat{\boldsymbol{V}}_{t+1}) - vec(\boldsymbol{V})}_2 \notag \\
 &= \norm{(\mathcal{X}(\hat{\mathcal{J}})^{\top}\mathcal{X}(\hat{\mathcal{J}}))^{-1}\mathcal{X}(\hat{\mathcal{J}})^{\top}\mathcal{R} - vec(\boldsymbol{V})}_2 \notag \\
 &\leq\frac{1}{\lambda t}(\norm{\mathcal{X}(\hat{\mathcal{J}})^{\top}\eta}_2 + \norm{\boldsymbol{Q}_{\theta,t}}_2), \label{eq:norm-v}
\end{align}
where the last inequality is because $\lambda_{\min}(\widehat{\overline{\boldsymbol{\Sigma}}}_{t,\hat{J}}) \geq \lambda$ and $\boldsymbol{Q}_{\theta,t} = \sum_{s=1}^tvec(\mathring{\boldsymbol{X}}_{a_s,u_s})vec(\mathring{\boldsymbol{X}}_{a_s,u_s}\boldsymbol{W}^{\top}) (vec(\widehat{\boldsymbol{\Theta}}_s) - vec(\boldsymbol{\Theta}))$.

Similarly, according to Lemma~\ref{lem:als-est-err} and Hoeffding's inequality, with probability at least $1-\delta$,
\begin{align*}
    \norm{\boldsymbol{Q}_{\theta,t}}_2 \leq 2s_xs_{\theta}\sqrt{2t\log(2d/\delta)}.
\end{align*}

Putting all these together, we have
\begin{align*}
    &\mathbb{E}\Big[\norm{vec(\boldsymbol{V}) - vec(\widehat{\boldsymbol{V}}_t)}_2| \mathcal{D}_t,\mathcal{L}_t^{\frac{\beta}{4\nu\omega_{\mathcal{X}}}}\Big] \leq \notag \\
    & \mathbb{E}\Big[\frac{1}{\lambda t}\Big(\norm{\mathcal{X}(\hat{\mathcal{J}})^{\top}\eta}_2  + 2s_xs_{\theta}\sqrt{2t\log(2d/\delta)} \Big)| \mathcal{D}_t,\mathcal{L}_t^{\frac{\beta}{4\nu\omega_{\mathcal{X}}}} \Big],
\end{align*}
where $\lambda = \lambda_{\min}(\widehat{\overline{\boldsymbol{\Sigma}}}_{t,\hat{\mathcal{J}}}) \geq \frac{\beta}{4\nu\omega_{\mathcal{X}}}$.

For the first term, denote the event $\mathcal{I}_h = \Big \{\frac{1}{t}\norm{\mathcal{X}(\hat{J})^{\top}\eta}_2 \in (2h\delta, 2\delta(h+1)] \Big\}$, where
\begin{gather}
    \delta = \frac{\sigma s_A\nu \omega_{\mathcal{X}}}{\beta}\sqrt{\frac{32(s_0+\frac{4\nu\omega_{\mathcal{X}}}{\phi_0^2})}{t-1}},
\end{gather}
then
\begin{align*}
    &\mathbb{E}\Big[\frac{1}{\lambda t}\norm{\mathcal{X}(\hat{\mathcal{J}})^{\top}\eta}_2| \mathcal{D}_t,\mathcal{L}_t^{\frac{\beta}{4\nu\omega_{\mathcal{X}}}} \Big] \notag \\
    &\leq \sum_{h=0}^{\lceil s_1/\delta \rceil}\delta (h+1)\mathbb{P}\Big( \norm{\mathcal{X}(\hat{\mathcal{J}})^{\top}\eta}_2 \geq 2th\delta |\mathcal{I}_h, \mathcal{D}_t,\mathcal{L}_t^{\frac{\beta}{4\nu\omega_{\mathcal{X}}}} \Big).
\end{align*}
Denote $s' = s_0 + \frac{4\nu\omega_{\mathcal{X}}\sqrt{s_0}}{\phi_0^2}$. Then, using Lemma~\ref{lem:eat}, we get
\begin{align}
    &\mathbb{P}\Big ( \norm{\mathcal{X}(\hat{\mathcal{J}})^{\top}\eta}_2 \geq 2th\delta |\mathcal{I}_h, \mathcal{D}_t,\mathcal{L}_t^{\frac{\beta}{4\nu\omega_{\mathcal{X}}}} \Big) \notag \\
    &\leq 2s'\exp(-\frac{\beta^2t\delta^2h^2}{32\sigma^2s_A^2\nu\omega_{\mathcal{X}}^2s'}) = 2s'\exp(-h^2).
\end{align}
We also trivially have that:
\begin{align}
    \mathbb{P}\Big ( \norm{\mathcal{X}(\hat{J})^{\top}\eta}_2 \geq 2th\delta |\mathcal{I}_h, \mathcal{D}_t,\mathcal{L}_t^{\frac{\beta}{4\nu\omega_{\mathcal{X}}}} \Big) \leq 1.
\end{align}
Thus, we have
\begin{align}
    &\mathbb{E}\Big[\frac{1}{\lambda t}\norm{\mathcal{X}(\hat{\mathcal{J}})^{\top}\eta}_2| \mathcal{D}_t,\mathcal{L}_t^{\frac{\beta}{4\nu\omega_{\mathcal{X}}}} \Big] \notag \\
    &\leq \sum_{h=0}^{\lceil s_1/\delta \rceil}\delta (h+1)\mathbb{P}\Big( \norm{\mathcal{X}(\hat{\mathcal{J}})^{\top}\eta}_2 \geq 2th\delta |\mathcal{I}_h, \mathcal{D}_t,\mathcal{L}_t^{\frac{\beta}{4\nu\omega_{\mathcal{X}}}} \Big) \notag \\
    &\leq  \Big(\sum_{h=0}^{\lceil s_1/\delta \rceil}\delta (h+1) \min \{1,2s'\exp(-h^2)\}\Big) \notag \\
    &\leq \Big(\sum_{h=0}^{h_0}\delta(h+1)+\sum_{h = h_0+1}^{h_{\max}}\delta 2s'(h+1) \exp(-h^2) \Big),
\end{align}
where we set $h_0 = \lfloor \sqrt{\log 4s'} + 1\rfloor$. Based on the proof of Lemma 5.9 in \cite{ariu2022thresholded}, we can obtain
\begin{align}
    &\sum_{h=0}^{h_0}(h+1)+\sum_{h = h_0+1}^{h_{\max}}2s'(h+1) \exp(-h^2) \notag \\
    &\leq \frac{(h_0+1)(h_0+2)}{2} + 2s'\exp(-h_0^2) \leq \frac{7h_0^2}{2}.
\end{align}
Then we have
\begin{align*}
    &\mathbb{E}\Big[\norm{vec(\boldsymbol{V}) - vec(\widehat{\boldsymbol{V}}_t)}_2| \mathcal{D}_t,\mathcal{L}_t^{\frac{\beta}{4\nu\omega_{\mathcal{X}}}}\Big] \notag \\
    & \leq 14h_0^2 \frac{\sigma s_A\nu \omega_{\mathcal{X}}}{\beta}\sqrt{\frac{2(s_0+\frac{4\nu\omega_{\mathcal{X}}}{\phi_0^2})}{t-1}} \notag \\
    &\quad + \frac{4\nu\omega_{\mathcal{X}}}{\beta}2s_xs_{\theta}\sqrt{\frac{2}{t}\log(2d/\delta)}.
\end{align*}
In summary, we can bound the instantaneous regret
\begin{align}
    &\mathbb{E}[R_t] \leq \frac{4\nu\omega_{\mathcal{X}}}{\beta}2s_xs_{\theta}\sqrt{\frac{2}{t}\log(2d/\delta)}   \notag \\
    &\quad+ \frac{14h_0^2\sigma s_A\nu \omega_{\mathcal{X}}}{\beta}\sqrt{\frac{2(s_0+\frac{4\nu\omega_{\mathcal{X}}}{\phi_0^2})}{t-1}} \notag \\
    &\quad  + \alpha_t\norm{vec(\mathring{\boldsymbol{X}}_{a_t,u_t}\boldsymbol{W}^{\top})}_{\boldsymbol{A}_{t-1}^{-1}} + \alpha_t  \norm{vec(\mathring{\boldsymbol{X}}_{a_t^*,u_t}\boldsymbol{W}^{\top})}_{\boldsymbol{A}_{t-1}^{-1}} \notag \\
    &\quad  + 2s_1s_A\Big (\mathbb{P}(\mathcal{D}_t^c) + \mathbb{P}((\mathcal{L}_t^{\frac{\beta}{4\nu\omega_{\mathcal{X}}}})^c | \mathcal{D}_t) \Big). 
\end{align}
\end{proof}
%-------------------------------

%------------------------------>Proof
\subsection{Proof of Theorem~\ref{the:rgt}}\label{app:the}
\begin{proof}
\begin{align}
   R(T) \leq 2s_As_1\tau + \text{term (a)} + \text{term (b)} + \text{term (c)},  \label{eq:rgt-term}
\end{align}
where 
\begin{align}
    \text{term (a)} &= \sum_{t=1}^T \Big (2\alpha_t\norm{vec(\mathring{\boldsymbol{X}}_{a_t,u_t}\boldsymbol{W}^{\top})}_{\boldsymbol{A}_{t-1}^{-1}} \notag \\
    &\quad + \alpha_t  \norm{vec(\mathring{\boldsymbol{X}}_{a_t^*,u_t}\boldsymbol{W}^{\top})}_{\boldsymbol{A}_{t-1}^{-1}} \Big), \\
    \text{term (b)} &= \sum_{t=1}^T \Big(\frac{4\nu\omega_{\mathcal{X}}}{\beta}2s_xs_{\theta}\sqrt{\frac{2}{t}\log(2d/\delta)} \notag \\
    &\quad + \frac{14h_0^2\sigma s_A\nu \omega_{\mathcal{X}}}{\beta}\sqrt{\frac{2(s_0+\frac{4\nu\omega_{\mathcal{X}}}{\phi_0^2})}{t-1}} \Big), 
\end{align}
\begin{align}
    \text{term (c)} &= \sum_{t=\tau+1}^T 2s_1s_A\Big (\mathbb{P}(\mathcal{D}_t^c) + \mathbb{P}((\mathcal{L}_t^{\frac{\beta}{4\nu\omega_{\mathcal{X}}}})^c | \mathcal{D}_t) \Big).
\end{align}

According to Lemma 11 in \cite{abbasi2011improved},
\begin{align}
    &\sum_{t=1}^T\norm{vec(\mathring{\boldsymbol{X}}_{a_t,u_t}\boldsymbol{W}^{\top})}_{\boldsymbol{A}_{t-1}^{-1}}^2 \leq 2\log(\frac{\det(\boldsymbol{A}_T)}{\det(\boldsymbol{A}_0)}) \label{eq:b1} 
\end{align}
Then we can derive
\begin{align*}
    \text{term (a)} 
    &\leq 3\alpha_T\sqrt{2Tdn \log(1+\frac{\sum_{t=1}^T\sum_{j=1}^n \boldsymbol{W}(u_t,j)^2}{dn \lambda_1})},
\end{align*}
Then we bound term (b) of \eqref{eq:rgt-term}, 
\begin{align*}
    &\quad \quad \text{term (b)} \\
    %&\leq \Big(\frac{7h_0^2\sigma}{2}\sqrt{2(s_0+\frac{4\nu\omega_{\mathcal{X}}}{\phi_0^2})} + 2s_xs_{\theta}\sqrt{2\log(2d/\delta)}\Big) \times \notag \\
    %&\quad \quad \frac{4s_A\nu\omega_{\mathcal{X}}}{\beta}(1+\int_1^T\sqrt{\frac{1}{t}}dt) \notag \\
    &\leq  \frac{4s_A\nu\omega_{\mathcal{X}}}{\beta} \Big(7h_0^2\sigma \sqrt{2(s_0+\frac{4\nu\omega_{\mathcal{X}}}{\phi_0^2})T}  + 4s_xs_{\theta} \sqrt{T\log(2d/\delta)}\Big).
\end{align*}
According to the proof in \cite{ariu2022thresholded}, with $vec(\boldsymbol{V})$'s dimension $nd$ and $\tau = \lfloor \frac{2\log(2n^2d^2)}{C_0^2(\log s_0)(\log \log nd)} \rfloor$, we have
\begin{align}
    \sum_{t=\tau + 1}^T \exp(-\frac{t\lambda_t^2}{32\sigma^2s_A^2}+\log(nd)) \leq \frac{\pi^2}{6}. 
\end{align}

Thus, we can bound term (c) of \eqref{eq:rgt-term}
\begin{align}
    \text{term (c)} &\leq \frac{\pi^2}{6} + \frac{2}{C_0^2} + \Big( s_0 + \frac{4\nu\omega_{\mathcal{X}}\sqrt{s_0}}{\phi_0^2} \Big)^2\frac{20s_A^2\nu\omega_{\mathcal{X}}}{\beta}.
\end{align}
Putting all of these together, we can bound the regret
\begingroup
\footnotesize
\begin{align}
    &R(T) \leq 2s_As_1\tau + \frac{\pi^2}{6} + \frac{2}{C_0^2}  + \Big( s_0 + \frac{4\nu\omega_{\mathcal{X}}\sqrt{s_0}}{\phi_0^2} \Big)^2\frac{20s_A^2\nu\omega_{\mathcal{X}}}{\beta} \notag \\
    &\quad \quad + \frac{28h_0^2\sigma s_A\nu\omega_{\mathcal{X}}}{\beta}  \sqrt{2(s_0+\frac{4\nu\omega_{\mathcal{X}}}{\phi_0^2})T} \notag \\
    &\quad \quad + \frac{16s_A\nu\omega_{\mathcal{X}}}{\beta} \sqrt{T\log(2d/\delta)} \Big) \notag \\
    &\quad \quad  + 3\alpha_{T_r}\sqrt{2Tdn \log(1+\frac{\sum_{t=1}^T\sum_{j=1}^n \boldsymbol{W}(u_t,j)^2}{dn \lambda_1})},
\end{align}
\endgroup
with $\alpha_T$ defined in Lemma~\ref{lem:ucb}.
\end{proof}
%---------------------------------

\subsection{Proof of Theorem~\ref{the:lb}}\label{app:the-lb}
\begin{proof}
The proof combines two known lower bounds from the literature: Collaborative linear bandit regret and sparse anomaly component. In the following, we analyze these two parts separately. 

In the best-case scenario (all users are identical and fully connected), the regret of collaborative linear bandit is lower bounded by $\Omega(d\sqrt{nT})$, while in the worst-case scenario (all users are completely isolated and different), it is lower bounded by $\Omega(dn\sqrt{T})$. Therefore, the regret of network learning is lower bounded by $\Omega(\xi d\sqrt{nT})$, where $\xi \in [1,\sqrt{n}]$ where $\xi$ serves as a scaling factor jointly determined by the connectivity of the user graph and the similarity among users’ parameters. 

In the sparse linear bandit setting (where an unknown sparse corruption vector affects the rewards), previous work \cite{ren2024dynamic} has proven that regret is lower bounded by $\Omega(\sqrt{s_0T})$. 

By combining these two challenges, i.e. learning across a network and handling sparse anomalies, we conclude that any algorithm must suffer at least $\Omega(\xi d\sqrt{nT} + \sqrt{s_0T})$ regret.
\end{proof}

\bibliographystyle{IEEEtran}
\bibliography{references}  

\begin{thebibliography}{10}
\providecommand{\url}[1]{#1}
\csname url@samestyle\endcsname
\providecommand{\newblock}{\relax}
\providecommand{\bibinfo}[2]{#2}
\providecommand{\BIBentrySTDinterwordspacing}{\spaceskip=0pt\relax}
\providecommand{\BIBentryALTinterwordstretchfactor}{4}
\providecommand{\BIBentryALTinterwordspacing}{\spaceskip=\fontdimen2\font plus
\BIBentryALTinterwordstretchfactor\fontdimen3\font minus \fontdimen4\font\relax}
\providecommand{\BIBforeignlanguage}[2]{{%
\expandafter\ifx\csname l@#1\endcsname\relax
\typeout{** WARNING: IEEEtran.bst: No hyphenation pattern has been}%
\typeout{** loaded for the language `#1'. Using the pattern for}%
\typeout{** the default language instead.}%
\else
\language=\csname l@#1\endcsname
\fi
#2}}
\providecommand{\BIBdecl}{\relax}
\BIBdecl

\bibitem{li2010exploitation}
W.~Li, X.~Wang, R.~Zhang, Y.~Cui, J.~Mao, and R.~Jin, ``Exploitation and exploration in a performance based contextual advertising system,'' in \emph{Proceedings of the 16th ACM SIGKDD International Conference on Knowledge Discovery and Data Mining}, 2010, pp. 27--36.

\bibitem{mahadik2020fast}
K.~Mahadik, Q.~Wu, S.~Li, and A.~Sabne, ``Fast distributed bandits for online recommendation systems,'' in \emph{Proceedings of the 34th ACM International Conference on Supercomputing}, 2020, pp. 1--13.

\bibitem{wu2016contextual}
Q.~Wu, H.~Wang, Q.~Gu, and H.~Wang, ``Contextual bandits in a collaborative environment,'' in \emph{Proceedings of the 39th International ACM SIGIR conference on Research and Development in Information Retrieval}, 2016, pp. 529--538.

\bibitem{bouneffouf2012contextual}
D.~Bouneffouf, A.~Bouzeghoub, and A.~L. Gan{\c{c}}arski, ``A contextual-bandit algorithm for mobile context-aware recommender system,'' in \emph{Neural Information Processing: 19th International Conference, ICONIP 2012, Doha, Qatar, November 12-15, 2012, Proceedings, Part III 19}.\hskip 1em plus 0.5em minus 0.4em\relax Springer, 2012, pp. 324--331.

\bibitem{cesa2013gang}
N.~Cesa-Bianchi, C.~Gentile, and G.~Zappella, ``A gang of bandits,'' \emph{Advances in Neural Information Processing Systems}, vol.~26, 2013.

\bibitem{yang2020laplacian}
K.~Yang, L.~Toni, and X.~Dong, ``Laplacian-regularized graph bandits: Algorithms and theoretical analysis,'' in \emph{International Conference on Artificial Intelligence and Statistics}.\hskip 1em plus 0.5em minus 0.4em\relax PMLR, 2020, pp. 3133--3143.

\bibitem{li2017radar}
J.~Li, H.~Dani, X.~Hu, and H.~Liu, ``Radar: Residual analysis for anomaly detection in attributed networks.'' in \emph{IJCAI}, vol.~17, 2017, pp. 2152--2158.

\bibitem{10909361}
X.~Cheng, B.~Nourani-Koliji, and S.~Maghsudi, ``Online influence maximization with semi-bandit feedback under corruptions,'' \emph{IEEE Transactions on Network Science and Engineering}, vol.~12, no.~3, pp. 2308--2321, 2025.

\bibitem{zhou2014detection}
W.~Zhou, Y.~S. Koh, J.~Wen, S.~Alam, and G.~Dobbie, ``Detection of abnormal profiles on group attacks in recommender systems,'' in \emph{Proceedings of the 37th international ACM SIGIR conference on Research \& development in information retrieval}, 2014, pp. 955--958.

\bibitem{grun2015single}
D.~Gr{\"u}n, A.~Lyubimova, L.~Kester, K.~Wiebrands, O.~Basak, N.~Sasaki, H.~Clevers, and A.~Van~Oudenaarden, ``Single-cell messenger rna sequencing reveals rare intestinal cell types,'' \emph{Nature}, vol. 525, no. 7568, pp. 251--255, 2015.

\bibitem{chaudhary2015folding}
P.~Chaudhary, A.~N. Naganathan, and M.~M. Gromiha, ``Folding race: a robust method for predicting changes in protein folding rates upon point mutations,'' \emph{Bioinformatics}, vol.~31, no.~13, pp. 2091--2097, 2015.

\bibitem{bolton2001unsupervised}
R.~J. Bolton, D.~J. Hand \emph{et~al.}, ``Unsupervised profiling methods for fraud detection,'' \emph{Credit Scoring and Credit Control VII}, pp. 235--255, 2001.

\bibitem{kshetri2010economics}
N.~Kshetri, ``The economics of click fraud,'' \emph{IEEE Security \& Privacy}, vol.~8, no.~3, pp. 45--53, 2010.

\bibitem{schiff2017screening}
G.~D. Schiff, L.~A. Volk, M.~Volodarskaya, D.~H. Williams, L.~Walsh, S.~G. Myers, D.~W. Bates, and R.~Rozenblum, ``Screening for medication errors using an outlier detection system,'' \emph{Journal of the American Medical Informatics Association}, vol.~24, no.~2, pp. 281--287, 2017.

\bibitem{chandola2009anomaly}
V.~Chandola, A.~Banerjee, and V.~Kumar, ``Anomaly detection: A survey,'' \emph{ACM Computing Surveys (CSUR)}, vol.~41, no.~3, pp. 1--58, 2009.

\bibitem{ding2019interactive}
K.~Ding, J.~Li, and H.~Liu, ``Interactive anomaly detection on attributed networks,'' in \emph{Proceedings of the Twelfth ACM International Conference on Web Search and Data Mining}, 2019, pp. 357--365.

\bibitem{laxhammar2013online}
R.~Laxhammar and G.~Falkman, ``Online learning and sequential anomaly detection in trajectories,'' \emph{IEEE Transactions on Pattern Analysis and Machine Intelligence}, vol.~36, no.~6, pp. 1158--1173, 2013.

\bibitem{gentile2014online}
C.~Gentile, S.~Li, and G.~Zappella, ``Online clustering of bandits,'' in \emph{International Conference on Machine Learning}.\hskip 1em plus 0.5em minus 0.4em\relax PMLR, 2014, pp. 757--765.

\bibitem{li2016collaborative}
S.~Li, A.~Karatzoglou, and C.~Gentile, ``Collaborative filtering bandits,'' in \emph{Proceedings of the 39th International ACM SIGIR conference on Research and Development in Information Retrieval}, 2016, pp. 539--548.

\bibitem{li2019improved}
S.~Li, W.~Chen, S.~Li, and K.-S. Leung, ``Improved algorithm on online clustering of bandits,'' in \emph{Proceedings of the 28th International Joint Conference on Artificial Intelligence}, 2019, pp. 2923–--2929.

\bibitem{cheng2023parallel}
X.~Cheng, C.~Pan, and S.~Maghsudi, ``Parallel online clustering of bandits via hedonic game,'' in \emph{International Conference on Machine Learning}.\hskip 1em plus 0.5em minus 0.4em\relax PMLR, 2023, pp. 5485--5503.

\bibitem{zhuang2017identifying}
H.~Zhuang, C.~Wang, and Y.~Wang, ``Identifying outlier arms in multi-armed bandit,'' \emph{Advances in Neural Information Processing Systems}, vol.~30, 2017.

\bibitem{ban2020generic}
Y.~Ban and J.~He, ``Generic outlier detection in multi-armed bandit,'' in \emph{Proceedings of the 26th ACM SIGKDD International Conference on Knowledge Discovery and Data Mining}, 2020, pp. 913--923.

\bibitem{zhu2020robust}
Y.~Zhu, S.~Katariya, and R.~Nowak, ``Robust outlier arm identification,'' in \emph{International Conference on Machine Learning}.\hskip 1em plus 0.5em minus 0.4em\relax PMLR, 2020, pp. 11\,566--11\,575.

\bibitem{meng2021interactive}
X.~Meng, Y.~Wang, S.~Wang, D.~Yao, and Y.~Zhang, ``Interactive anomaly detection in dynamic communication networks,'' \emph{IEEE/ACM Transactions on Networking}, vol.~29, no.~6, pp. 2602--2615, 2021.

\bibitem{zhang2023multi}
W.~Zhang, X.~Meng, J.~Li, Y.~Wang, and Y.~Zhang, ``Multi-layer collaborative bandit for multivariate time series anomaly detection,'' in \emph{2023 IEEE/ACM 31st International Symposium on Quality of Service (IWQoS)}.\hskip 1em plus 0.5em minus 0.4em\relax IEEE, 2023, pp. 1--10.

\bibitem{perozzi2014focused}
B.~Perozzi, L.~Akoglu, P.~Iglesias~S{\'a}nchez, and E.~M{\"u}ller, ``Focused clustering and outlier detection in large attributed graphs,'' in \emph{Proceedings of the 20th ACM SIGKDD International Conference on Knowledge Discovery and Data Mining}, 2014, pp. 1346--1355.

\bibitem{liu2017accelerated}
N.~Liu, X.~Huang, and X.~Hu, ``Accelerated local anomaly detection via resolving attributed networks.'' in \emph{IJCAI}, 2017, pp. 2337--2343.

\bibitem{peng2018anomalous}
Z.~Peng, M.~Luo, J.~Li, H.~Liu, Q.~Zheng \emph{et~al.}, ``Anomalous: A joint modeling approach for anomaly detection on attributed networks.'' in \emph{IJCAI}, vol.~18, 2018, pp. 3513--3519.

\bibitem{wang2017factorization}
H.~Wang, Q.~Wu, and H.~Wang, ``Factorization bandits for interactive recommendation,'' in \emph{Proceedings of the AAAI Conference on Artificial Intelligence}, vol.~31, no.~1, 2017.

\bibitem{zhao2016online}
R.~Zhao and V.~Y. Tan, ``Online nonnegative matrix factorization with outliers,'' \emph{IEEE Transactions on Signal Processing}, vol.~65, no.~3, pp. 555--570, 2016.

\bibitem{liu2018transferable}
B.~Liu, Y.~Wei, Y.~Zhang, Z.~Yan, and Q.~Yang, ``Transferable contextual bandit for cross-domain recommendation,'' in \emph{Proceedings of the AAAI Conference on Artificial Intelligence}, vol.~32, no.~1, 2018.

\bibitem{ariu2022thresholded}
K.~Ariu, K.~Abe, and A.~Prouti{\`e}re, ``Thresholded lasso bandit,'' in \emph{International Conference on Machine Learning}.\hskip 1em plus 0.5em minus 0.4em\relax PMLR, 2022, pp. 878--928.

\bibitem{oh2021sparsity}
M.-h. Oh, G.~Iyengar, and A.~Zeevi, ``Sparsity-agnostic lasso bandit,'' in \emph{International Conference on Machine Learning}.\hskip 1em plus 0.5em minus 0.4em\relax PMLR, 2021, pp. 8271--8280.

\bibitem{cella2023multi}
L.~Cella, K.~Lounici, G.~Pacreau, and M.~Pontil, ``Multi-task representation learning with stochastic linear bandits,'' in \emph{International Conference on Artificial Intelligence and Statistics}.\hskip 1em plus 0.5em minus 0.4em\relax PMLR, 2023, pp. 4822--4847.

\bibitem{abbasi2011improved}
Y.~Abbasi-Yadkori, D.~P{\'a}l, and C.~Szepesv{\'a}ri, ``Improved algorithms for linear stochastic bandits,'' \emph{Advances in Neural Information Processing Systems}, vol.~24, 2011.

\bibitem{bennett2007netflix}
J.~Bennett, S.~Lanning \emph{et~al.}, ``The netflix prize,'' in \emph{Proceedings of KDD Cup and Workshop}, vol. 2007.\hskip 1em plus 0.5em minus 0.4em\relax New York, 2007, p.~35.

\bibitem{silva2023exploring}
N.~Silva, T.~Silva, H.~Hott, Y.~Ribeiro, A.~Pereira, and L.~Rocha, ``Exploring scenarios of uncertainty about the users' preferences in interactive recommendation systems,'' in \emph{Proceedings of the 46th International ACM SIGIR Conference on Research and Development in Information Retrieval}, 2023, pp. 1178--1187.

\bibitem{tropp2011user}
J.~A. Tropp, ``User-friendly tail bounds for matrix martingales,'' \emph{ACM Report}, vol.~1, 2011.

\bibitem{buhlmann2011statistics}
P.~B{\"u}hlmann and S.~Van De~Geer, \emph{Statistics for high-dimensional data: methods, theory and applications}.\hskip 1em plus 0.5em minus 0.4em\relax Springer Science \& Business Media, 2011.

\bibitem{wainwright2019high}
M.~J. Wainwright, \emph{High-dimensional statistics: A non-asymptotic viewpoint}.\hskip 1em plus 0.5em minus 0.4em\relax Cambridge University Press, 2019, vol.~48.

\bibitem{li2022asynchronous}
C.~Li and H.~Wang, ``Asynchronous upper confidence bound algorithms for federated linear bandits,'' in \emph{International Conference on Artificial Intelligence and Statistics}.\hskip 1em plus 0.5em minus 0.4em\relax PMLR, 2022, pp. 6529--6553.

\bibitem{ren2024dynamic}
Z.~Ren and Z.~Zhou, ``Dynamic batch learning in high-dimensional sparse linear contextual bandits,'' \emph{Management Science}, vol.~70, no.~2, pp. 1315--1342, 2024.

\end{thebibliography}

\end{document}